\documentclass[journal,draftcls,onecolumn,letterpaper,12pt]{IEEEtran}
\usepackage{cite,amsmath,amssymb,graphicx}
\usepackage{citesort}
\usepackage{algorithm,algorithmic}
\usepackage{subfigure}
\usepackage{mathrsfs}
\usepackage{url}
\usepackage{bbm}
\usepackage{amsthm}
\usepackage{setspace}
\usepackage{color}
\newtheorem{theorem}{Theorem}

\newtheorem{condition}{Condition}

\newtheorem{definition}{Definition}

\newtheorem{proposition}{Proposition}

\newtheorem{lemma}{Lemma}

\newtheorem{remark}{Remark}

\makeatletter
\newcounter{parentalgorithm}

\newcommand{\Prob}{\mathbb{P}}
\newcommand{\bL}{\boldsymbol{L}}
\newcommand{\bJ}{\boldsymbol{J}}
\newcommand{\bS}{\boldsymbol{S}}
\newcommand{\bEta}{\boldsymbol{\eta}}
\newcommand{\bW}{\boldsymbol{W}}

\newcommand{\bP}{\boldsymbol{P}}

\newcommand{\bD}{\boldsymbol{D}}

\newcommand{\bA}{\boldsymbol{A}}
\newcommand{\E}{\mathbb{E}}
\newcommand{\T}{\mathsf{T}}

\newcommand{\cD}{\mathcal{D}}

\newcommand{\cV}{\mathcal{V}}
\newcommand{\cM}{\mathcal{M}}

\newcommand{\cG}{\mathcal{G}}

\newcommand{\cE}{\mathcal{E}}
\newcommand{\cH}{\mathcal{H}}

\newcommand{\cN}{\mathcal{N}}

\newcommand{\by}{{\boldsymbol{y}}}
\newcommand{\bx}{{\boldsymbol{x}}}
\newcommand{\bs}{{\boldsymbol{s}}}
\newcommand{\be}{{\boldsymbol{e}}}
\newcommand{\bv}{{\boldsymbol{v}}}
\newcommand{\bu}{{\boldsymbol{u}}}

\newcommand{\ignore}[1]{{}}

\newcommand{\lb}{\left(}
\newcommand{\rb}{\right)}
\begin{document}
\title{Order-$2$ Asymptotic Optimality of the Fully Distributed Sequential Hypothesis Test}
\author{Shang Li and Xiaodong Wang
}
\maketitle
\vspace*{-12mm}
\begin{abstract}

This work analyzes the asymptotic performances of fully distributed sequential hypothesis testing procedures as the type-I and type-II error rates approach zero, in the context of a sensor network without a fusion center. In particular, the sensor network is defined by an undirected graph, where each sensor can observe samples over time, access the information from the adjacent sensors, and perform the sequential test based on its own decision statistic. Different from most literature, the sampling process and the information exchange process in our framework take place simultaneously (or, at least in comparable time-scales), thus cannot be decoupled from one another. Our goal is to achieve order-$2$ asymptotically optimal performance at every sensor, i.e., the average detection delay is within a constant gap from the centralized optimal sequential test as the error rates approach zero. To that end, two message-passing schemes are considered, based on which the distributed sequential probability ratio test (DSPRT) is carried out respectively. The first scheme features the dissemination of the raw samples. In specific, every sample propagates over the network by being relayed from one sensor to another until it reaches all the sensors in the network. Although the sample propagation based DSPRT is shown to yield the asymptotically optimal performance at each sensor, it incurs excessive inter-sensor communication overhead due to the exchange of raw samples with index information. The second scheme adopts the consensus algorithm, where the local decision statistic is exchanged between sensors instead of the raw samples, thus significantly lowering the communication requirement compared to the first scheme. In particular, the decision statistic for DSPRT at each sensor is updated by the weighted average of the decision statistics in the neighbourhood at every message-passing step. We show that, under certain regularity conditions, the consensus algorithm based DSPRT also yields the order-$2$ asymptotically optimal performance at all sensors. Our asymptotic analyses of the two message-passing based DSPRTs are then corroborated by simulations using the Gaussian and Laplacian samples. 

\ignore{Distinct from many classic literature where sufficiently large number of messaging passings are implemented before the next sample arrives. That is, we consider the general case where $q$ message-passings are performed between two sampling instants, where $q=1$ corresponds to the ``running consensus'' protocol. We prove that, for fixed network size $K$, the message-passing based SPRT yields order-$2$ asymptotic optimality as the error probabilities go to zeros. For fixed network-wide Kullback-Leibler divergence, the consensus algorithm based SPRT yields order-$1$ asymptotic optimality as the network size $K$ grows large. Practical bounds for estimating the expected stopping time and error probabilities are also obtained. Extensive numerical results are demonstrated to validate the theoretical conclusions. }
\end{abstract}

\newpage

\section{Introduction}

Following the optimal stopping rule in the data acquisition process, the sequential hypothesis test is able to reduce the data sample size compared to its fixed-sample-size counterpart. The sequential framework is particularly essential for systems where data are acquired in real-time, and the decision latency is of critical importance. In particular, for the simple null versus simple alternative hypothesis testing, the sequential probability ratio test (SPRT)  attains the minimum expected sample sizes under both hypotheses \cite{Wald48,SeqA_book}. For example, it only requires one fourth of the sample sizes on average as that of the fixed-sample-size test for detecting the mean-shift of Gaussian samples  \cite{Poor_book}. 

 Meanwhile, the recent decade has witnessed the surge of smart devices that can be connected through wireless links and form cooperative networks, giving rise to the emerging Internet of Things (IoT). Some examples include the body network where wearable devices are connected to the smartphone for health monitoring, the vehicular Ad Hoc network (VANET) as part of the intelligent transportation system, and the social network that connects people through online friendship. Many applications pertaining to these examples involve choosing between two hypotheses with stringent requirements on the decision latency, necessitating solutions that can integrate the sequential hypothesis test into the cooperative networks. For instance, VANETs can cooperatively detect the road congestion in a timely fashion; or social networks can determine whether a restaurant is good or bad with the help of the so-called collective wisdom.  
 
There are primarily two prototypes of network architectures, depending on whether or not there exists a central processing unit, or fusion center. In the presence of a fusion center, the network features a hierarchical structure (c.f., Fig. \ref{fig:sys}-(a)), i.e., all sensors directly transmit data to the fusion center, where the data fusion and sequential test are performed. The body network mentioned above falls under this category, usually with the smartphone functioning as the fusion center. Other variants of the hierarchical network include trees and tandem networks \cite{Blum97}.  The main challenge associated with the hierarchical network arises from  the communication burden from sensors to the fusion center. There is a rich body of studies that aim to ameliorate the communication overhead while preserving the  collaborative performance of the sequential hypothesis test\cite{Veeravalli93,Tsitsiklis93,Veeravalli94,LiLi15,Mei08,Wang13,Fellouris11} and the sequential change-point detection \cite{SLi14_SG,SLi16_TIFS}. 
 \begin{figure}
\centering
\subfigure[Hierarchical system]{
\includegraphics[width=0.36\columnwidth]{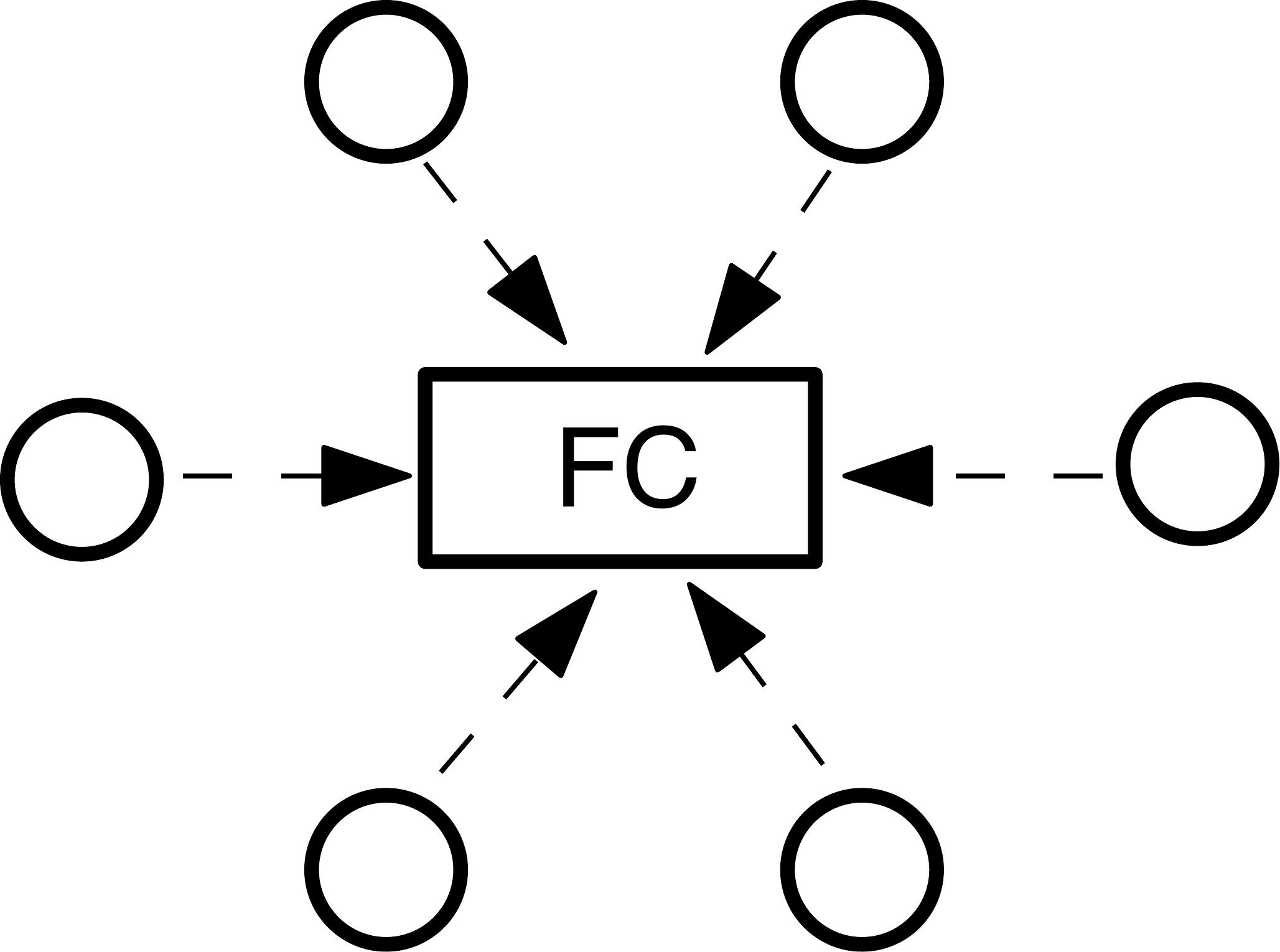}}\hspace*{12mm}
\subfigure[Distributed system]{
\includegraphics[width=0.36\columnwidth]{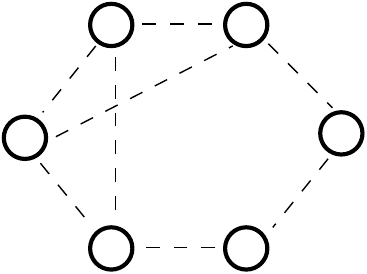}}
\caption{Illustration of the two types of sensor network architectures.}\label{fig:sys}
\end{figure}

In spite of its simple structure, the hierarchical network suffers from several limitations. First, it is susceptible to the fusion center malfunctioning. Second, it becomes very inefficient in the networks where there is no fusion center and every sensor needs to  function as a decision-maker. A typical example is the VANET, where each vehicle is able to make individual decision by exchanging data with other vehicles within its communication range. 
Accordingly, the distributed architecture (c.f., Fig. \ref{fig:sys}-(b)) is more natural and efficient in this case. In specific, the sensors are connected by wireless links, which allow  them to exchange data, and each sensor makes distributed decision based on its own available information. However, compared to the hierarchical network, the distributed network is prone to sub-optimal cooperative performance due to the lack of global information at each sensor. Therefore, the key challenge is to devise efficient information exchange mechanisms such that each sensor can optimize its distributed sequential test, and, if feasible, achieve the globally optimal performance. In this paper, we will consider two message-passing based distributed sequential tests, and prove their asymptotic optimalities. 
\subsection{Overview}
Since the seminal work by DeGroot \cite{DeGroot74}, the information aggregation in distributed networks has been widely studied. A majority of the existing literature builds on the fixed-sample-size paradigm. That is, each sensor starts with a private sample and aims to obtain the average of all private samples in the system (termed as ``reaching consensus'') through inter-sensor information exchange. The most popular information exchange protocols include the ``consensus algorithm'' and ``gossip algorithm'', whose comprehensive surveys can be found in  \cite{Saber07} and \cite{Dimakis10} respectively. More sophisticated scenario involving quantized message-passing and random link failures was investigated by \cite{Kar10}. In these works, a new sample is not allowed to enter into the network during the process of ``reaching consensus'', thus they are only relevant to the fixed-sample-size inference problems. 

\ignore{\cite{Xiao04}: Given the topology, design the weights such that the algebraic connectivity is maximized through a semi-definite convex programming. This leads to mild increase of algebraic connectivity. 
\cite{Saber05Conf}: Provided an approach that randomly generates the ``small-world'' network, which yields significantly larger algebraic connectivity giving rise to a ultrafast consensus rate. {\color{red} This could be used to generate our random experiment graph.}}

In contrast, the distributed sequential inference problem, where the complication arises from the successively arriving samples, is much less understood. Preliminarily, some existing works tackle this challenge by assuming that the consensus is reached before new samples are taken, which essentially decouples the sampling and the information aggregation processes and reduces the problem to the fixed-sample-size category \cite{ROSaber05,Khan08,Amini13,Ziyu16}. 
The more practical and interesting scenario is that the sampling and information aggregation processes take place simultaneously, or at least in comparable time-scales.  Under this setup, \cite{Kar13Mag} proposed the  ``consensus + innovation'' approach for distributed recursive parameter estimation; \cite{Braca08} intended to track a stochastic process using a ``running consensus'' algorithm. The same method was then  applied to the distributed locally optimal sequential test in \cite{Braca10}, where the alternative parameter is assumed to be close to the null one.
Moreover, the distributed sequential change-point detection was also investigated based on the concept of ``running consensus'' \cite{Stankovic11,BracaCD11,IlicGCD12}. 

While most of the above works focus on reaching (near) consensus on the value of local decision statistics, limited light has been shed upon the expected sample size, i.e., stopping time, and error probabilities of the distributed sequential test. Recently,  \cite{Srivastava14,Srivastava15Conf} analyzed the distributed sequential test  based on diffusion process (the continuous-time version of the consensus algorithm). For the  discrete-time model, \cite{Sahu16} used the ``consensus + innovation'' approach in combination with the sequential probability ratio test to detect the mean-shift of Gaussian samples. Closed-form bounds for the error probabilities and expected sample sizes of the distributed sequential test are derived. However, their analyses are restricted to one specific  testing problem, and do not reveal any asymptotic optimality.\subsection{Contributions}
 In this work, we consider two message-passing based distributed sequential tests. One requires the exchange of raw samples between adjacent sensors, while the other adopts the consensus algorithm as in \cite{Sahu16}. To the best of our knowledge, this work is the first to show the asymptotic optimality of a fully distributed sequential hypothesis test procedure. 
 Again, we emphsize that, due to the constantly arriving samples, reaching consensus on the values of the decision statistics at all sensors is generally impossible. Rather, our ultimate goal is to achieve the global (asymptotically) optimal  performance at every sensor in the network. 
In particular, the main contributions are summarized as follows.
\begin{itemize}
\item We consider a new distributed sequential test in Section III based on sample propagation, which allows each sample to reach other sensors as quickly as possible. This scheme is proved to achieve the order-$2$ asymptotically optimal performance at all sensors. 
\item We investigate the consensus-algorithm-based distributed sequential test for a generic hypothesis testing problem, whereas \cite{Sahu16} considered the particular problem of detecting the  Gaussian mean-shift. Moreover, we allow multiple rounds of message-passing between two sampling instants instead of one round as in \cite{Sahu16}.
\item We derive tighter analytical bounds to characterize the consensus-algorithm-based distributed sequential test, which leads to the order-$2$ asymptotic optimality. Our analyses also reveals that the constant gap to the optimal centralized performance can be reduced by increasing the number of message-passings between two adjacent sampling instants.
\end{itemize}
The remainder of the paper is organized as follows. Section II formulates the distributed sequential hypothesis testing problem. In Section III, we consider the distributed sequential test based on sample propagation and prove its asymptotic optimality. In Section IV, we prove the asymptotic optimality of the consensus-algorithm-based distributed sequential test. In Section V, simulation results based on Gaussian and Laplacian samples are given to corroborate the theoretical results. Finally, Section VI concludes the paper.  
\section{Problem Statement and Background}
Consider a network of $K$ sensors that sequentially take samples in parallel. Conditioned on the hypothesis, these samples are independent and identically distributed at each sensor and independent across sensors, i.e., 
\begin{align*}
&\cH_0: X^{(k)}_t\sim f_0^{(k)}(x), \\
&\cH_1: X^{(k)}_t\sim f_1^{(k)}(x), \quad k=1, 2, \ldots, K, \quad t=1, 2, \ldots
\end{align*}
The log-likelihood ratio (LLR) and the cumulative LLR up to time $t$ are denoted respectively as
\begin{align}
s_t^{(k)}\triangleq\log \underbrace{\frac{f^{(k)}_1(X^{(k)}_t)}{f^{(k)}_0(X^{(k)}_t)}}_{l^{(k)}_t}, \;\text{and}\; S_t^{(k)}\triangleq\sum_{j=1}^t s_j^{(k)}.
\end{align} 
The inter-sensor communication links determine the network topology, which can be represented by an undirected graph ${\cal G}\triangleq \{{\cal N}, {\cal E}\}$, with ${\cal N}$ being the set of sensors and ${\cal E}$ the set of edges.  In addition, let ${\cal N}_k$ be the set of neighbouring sensors that are directly connected to sensor $k$, i.e., $${\cal N}_k\triangleq \{j\in {\cal N}: \{k, j\}\in {\cal E}\}.$$ 
In distributed sequential test, at every time slot $t$ and each sensor $k$, the following actions take place in order: 1) taking a new sample, 2) exchanging messages with neighbours, and 3) deciding to stop for decision or to wait for more data at time $t+1$. Note that the first two actions, i.e., sampling and communication will continue even after the local test at sensor $k$ stops so that other sensors can still benefit from the same sample diversity, until all sensors stop. Mathematically, three components are to be designed for the distributed sequential test at each sensor:
\begin{itemize}
\item Exchanged messages: We denote the information transmitted from sensor $k$ to its adjacent sensors at time $t$ as $\cV_t^{(k)}$. In general, $\cV_t^{(k)}$ can be a set of numbers that depend on 
\begin{align}\label{im_constraint}
\left\{X_1^{(k)},\ldots, X_t^{(k)}, \left\{\cV_1^{(\ell)}\right\}_{\ell\in {\cal N}_k}, \ldots, \left\{\cV_{t-1}^{(\ell)}\right\}_{\ell\in {\cal N}_k}\right\}
\end{align} 
due to the distributed and causal assumptions. 


\item Stopping rule: The test stops for decision according to a stopping time random variable $\T$ that is adapted to the local information, i.e., 
\begin{align}\label{GeneralStopping}
\T^{(k)}\sim \left\{X_t^{(k)}, \left\{\cV_t^{(\ell)}\right\}_{\ell\in {\cal N}_k}\right\}_{t\in\mathbb{N}^+}. 
\end{align} 
Since we consider deterministic stopping rules, \eqref{GeneralStopping} means that $$\Prob\lb\left. \T^{(k)}\le t\right|X^{(k)}_1, \{\cV_1^{(\ell)}\}_{\ell\in {\cal N}_k}, \ldots, X_t^{(k)}, \{\cV_t^{(\ell)}\}_{\ell\in {\cal N}_k}\rb\in\{0, 1\}.$$ 
\item Decision function: Upon stopping at time $\T^{(k)}=t$, the terminal decision function chooses between the two hypotheses, i.e., 
\begin{align}\label{GeneralDecision}
D^{(k)}_{t}: \{X^{(k)}_1, \{\cV_1^{(\ell)}\}_{\ell\in {\cal N}_k}, \ldots, X_{t}^{(k)}, \{\cV_{t}^{(\ell)}\}_{\ell\in {\cal N}_k}\}\to \{0, 1\}.
\end{align}
For notational simplicity, we will omit the time index and use $D^{(k)}$ henthforth.
\end{itemize}
Accordingly, two performance metrics are used, namely, the expected stopping times $\E_i\T^{(k)}, \;i=0,1,$ and the type-I and type-II error probabilities, i.e., $\Prob_0\lb D^{(k)}=1\rb$ and $\Prob_1\lb D^{(k)}=0\rb$ respectively. The expected stopping times represent the average sample sizes under both hypotheses, and the error probabilities characterize the decision accuracy. As such, for the distributed sequential hypothesis testing, we aim to find the message design, stopping rule $\T^{(k)}$ and terminal decision function $D^{(k)}$ such that the expected stopping times at sensors under ${\cal H}_0$ and ${\cal H}_1$ are minimized subject to the error probability constraints:
\begin{align}\label{ST}
&\;\;\min_{\{\T^{(k)}, D^{(k)}, \{\cV_t^{(\ell)}\}_{\ell\in{\cal N}_k}\}}\;\quad\E_i\lb\T^{(k)}\rb, \quad i=0, 1\\
&\qquad\;\text{subject to} \quad\qquad \Prob_0\lb D^{(k)}=1\rb\le \alpha, \nonumber\\ &\phantom{\qquad subject to\quad\qquad}\; \Prob_1\lb D^{(k)}=0\rb\le \beta, \quad k = 1, 2, \ldots, K. \nonumber
\end{align}
Note that an implicit constraint in \eqref{ST} is given by the recursive definition of $V_t^{(k)}$ in \eqref{im_constraint}. Moreover, the above optimization is coupled across sensors due to the coupling of $V_t^{(k)}$.

Solving \eqref{ST} at the same time for $k=1, 2, \ldots, K$ is a formidable task except for  some special cases (for example, the fully connected network where all sensor pairs are connected, or the completely disconnected network where no two sensors are connected); therefore the asymptotically optimal solution is the next best thing to pursue.  We first introduce the widely-adopted definitions for the asymptotic optimalities \cite{Fellouris11}:
\begin{definition}\label{Def}
Let $\T^\star$ be the stopping time of the optimum sequential test that satisfies the two error probability constraints with equality. Then, as the Type-I and Type-II error probabilities $\alpha, \beta\to 0$, the sequential test that satisfies the error probability constraints with stopping time $\T$ is said to be order-$1$ asymptotically optimal if
\begin{align*}
1\le \frac{\E_i\lb\T\rb}{\E_i\lb\T^\star\rb}=1+o_{\alpha, \beta}(1);
\end{align*}
order-$2$ asymptotically optimal if
\begin{align*}
0\le {\E_i\lb\T\rb}-{\E_i\lb\T^\star\rb}=O(1).
\end{align*}
\end{definition}
Clearly, the order-$2$ asymptotic optimality is stronger than the order-$1$ asymptotic optimality since the  expected stopping time of the latter scheme can still diverge from the optimum, while the former scheme only deviates from the optimum by a constant as the error probabilities go to zero. 



Aiming at the asymptotically optimal solution, we start by finding a lower bound to \eqref{ST}. To this end, let us first consider the ideal case where the network is fully connected, i.e., ${\cal N}_k={\cal N}\,\backslash{\{k\}}$ for $k=1, 2, \ldots, K$. Then by setting $\cV_t^{(k)}=\{X_t^{(k)}\}, \; k=1, 2, \ldots, K$, every sensor can instantly obtain all data in the network, hence the system is equivalent to a centralized one. Consequently, given the error probability constraints, we can write
\begin{align}\label{optimalitybound}
\min_{\{\T^{(k)}, D^{(k)}, \{\cV_t^{(\ell)}\}_{\ell\in{\cal N}_k}\}}\;\E_i\lb\T^{(k)}\rb\ge \min_{\{\T^{(k)}, D^{(k)}, \{X_t^{(\ell)}\}_{\ell\in{\cal N}}\}}\;\E_i\lb\T^{(k)}\rb=\min_{\{\T, D\}}\;\lb\E_i\T\rb,
\end{align}
where $\T$ denotes the stopping time for the sequential test when all samples in the network are instantly available (referred to as the centralized setup). Naturally, invoking the classic result by \cite{Wald48}, $\min_{\{\T, D\}}\;\E_i\T$ in \eqref{optimalitybound} is solved with the centralized SPRT (CSPRT):
\begin{align}\label{CSPRT}
\T_c\triangleq \min \left\{t: {S}_t\triangleq\sum_{k=1}^KS_t^{(k)} \notin (-A, B)\right\},\quad D_c\triangleq \left\{
\begin{array}{ll}
1 & \text{if}\quad S_{\T_c}\ge B,\\
0 & \text{if}\quad S_{\T_c}\le -A,
\end{array}\right.
\end{align}
where $\{A,B\}$ are constants chosen such that the constraints in \eqref{ST} are satisfied with equalities. 
The asymptotic performance for the CSPRT as the error probabilities go to zero can be characterized by the following result \cite{SeqA_book}.
\begin{proposition}
The asymptotic performance of the CSPRT is characterized as
\begin{align}\label{GlobalOptimal}
\E_1\lb\T_c\rb=\frac{-\log\alpha}{\sum_{k=1}^K{\cal D}^{(k)}_1}, \quad \E_0\lb\T_c\rb=\frac{-\log\beta}{\sum_{k=1}^K{\cal D}^{(k)}_0}, \quad\text{as}\;\; \alpha, \beta \to 0,
\end{align}
where $\cD_i^{(k)}\triangleq \E_i\lb\log \frac{f^{(k)}_i(X)}{f^{(k)}_{1-i}(X)}\rb$ is the Kullback-Leibler divergence (KLD) at sensor $k$.
\end{proposition}
Proposition 1 gives the globally optimal performance that can only be achieved in the centralized step, whereas, in reality, the network is often a sparse one, far from being fully connected. Nevertheless, $\T_c$ will be used as a benchmark to evaluate our proposed distributed sequential tests in the next two sections. More specifically, by \eqref{optimalitybound}, we have 
\begin{align}\label{lowerbound}
\min_{\{\T^{(k)}, D^{(k)}, \{\cV_t^{(\ell)}\}_{\ell\in{\cal N}_k}\}}\;\E_i\lb\T^{(k)}\rb\ge \min_{\{\T, D\}}\;\E_i\lb\T\rb= \E_i\lb\T_c\rb;
\end{align} 
therefore, if any distributed sequential test attains the globally optimal performance given by \eqref{GlobalOptimal} in the sense defined by Definition \ref{Def} at all sensors, it is  asymptotically optimal.

A naive approach is to perform the local distributed SPRT (L-DSPRT), which adopts the same message-passing as the centralized test $\cV_t^{(k)}=\{X_t^{(k)}\}$. Hence the general definition of the stopping time in \eqref{GeneralStopping} becomes $\T^{(k)}\sim \{X_t^{(\ell)}, \ell\in \{{k,\cal N}_k\}\}_{t\in \mathbb{N}^+}$, i.e., the event $\{\T^{(k)}\le t\}$ (or its complementary event $\{\T^{(k)}> t\}$) only depends on $\{X_j^{(\ell)}, \ell\in \{{k,\cal N}_k\}\}_{j=1, \ldots, t}$, and the L-DSPRT is defined as
\begin{align}\label{LSPRT}
\T_\text{local}^{(k)}\triangleq \min \left\{t: \!\!\sum_{\ell\in \{k, {\cal N}_k\}}{S}^{(\ell)}_t \notin (-A, B)\right\},\;\; D_\text{local}^{(k)}\triangleq \left\{
\begin{array}{ll}
1 & \text{if}\quad  \sum_{\ell\in \{k, {\cal N}_k\}}{S}^{(\ell)}_{\T_\text{local}^{(k)}}\ge B,\\
0 & \text{if}\quad  \sum_{\ell\in \{k, {\cal N}_k\}}{S}^{(\ell)}_{\T_\text{local}^{(k)}}\le -A.
\end{array}\right.
\end{align}
Similarly, the asymptotic performance for L-DSPRT is readily obtained as
\begin{align}
\E_1\lb\T_\text{local}^{(k)}\rb=\frac{-\log\alpha}{\sum_{\ell\in \{k, {\cal N}_k\}}{\cal D}^{(\ell)}_1}, \quad \E_0\lb\T_\text{local}^{(k)}\rb=\frac{-\log\beta}{\sum_{\ell\in \{k, {\cal N}_k\}}{\cal D}^{(\ell)}_0}, \quad\text{as}\;\; \alpha, \beta \to 0.
\end{align}
Thus, compared with \eqref{GlobalOptimal}, $\T_\text{local}$ is sub-optimal in general, and may deviate substantially from the globally optimal performance, especially for the sensor with a small set of neighbours. 

In the next two sections, we will consider two message-passing-based distributed sequential tests, and show that they achieve order-$2$ asymptotic optimality (i.e., only deviate from \eqref{GlobalOptimal} by a constant), thus solving the distributed sequential hypothesis testing problem \eqref{ST} in the asymptotic regime where $\alpha, \beta\to 0$.


\section{Asymptotic Optimality of Distributed Sequential Test via Sample Dissemination}
In this section, we consider the first distributed sequential test based on sample dissemination. Simply put, in this scheme, every sample (or equivalently, the LLR of the sample) propagates through the network until it reaches all sensors. To some extent, it  resembles the scheme in \cite{DiLi15}, which, however, treats the message-passing and sequential test in decoupled manner. In our scheme, these two processes take place at the same time.

In order for the samples to reach all sensors, every new sample at one sensor needs to be relayed to the adjacent sensors at every message-passing step. These new samples include the newly collected sample and the external samples that come from the neighbours and have not been received before. \ignore{The sample that has been received before can be disregarded, since it is already stored and passed down to other sensors.} To implement this dissemination process, an implicit assumption is made that the samples are sent with index information such that they can be distinguished from one another. As indicated by the sub- and super-script of $s_t^{(k)}$, the index should include the sensor index $k$ that collects the sample and the time stamp $t$.  Overall, during the message-passing stage, each sensor needs to broadcast to its neighbours an array of messages, each of which is a sample with index information. 

To start with, we define two important quantities. The first is the information set ${\cal M}_t^{(k)}$ that contains all samples stored at sensor $k$ up to time $t$, which include  both local samples and external samples. For example, in set ${\cal M}_2^{(1)}=\left\{s_1^{(1)}, s_2^{(1)}, s_1^{(2)}, s_2^{(2)}, s_1^{(3)}\right\}$, $\{s_1^{(1)}, s_2^{(1)}\}$ are local samples, and $\{s_1^{(2)}, s_2^{(2)}, s_1^{(3)}\}$ are external samples  from sensors $2$ and $3$. The second is the message set ${\cal V}_t^{(k)}$ whose general form is given by \eqref{im_constraint}. In the sample dissemination scheme, they can be recursively updated as follows.
\begin{enumerate}
\item 
Sensor $k$ sends to the adjacent sensors the innovation $s_t^{(k)}$ and new external samples at last time $t-1$:
\begin{align}\label{sd1}
\cV_{t}^{(k)}\triangleq  \{s_{t}^{(k)}\}\cup \underbrace{\lb\cM_{t-1}^{(k)}- \cM_{t-2}^{(k)}- \{s_{t-1}^{(k)}\}\rb}_{\text{New external samples}}
\end{align}
where $\mathcal{A}- \mathcal{B}$ denotes the complementary set to $\mathcal{B}$ in $\mathcal{A}$.
\item  Sensor $k$ updates its information set with the innovation $s_t^{(k)}$ and  the messages from its neighbours, i.e, $\cup_{\ell\in \cN_k} \cV_t^{(\ell)}$:
\begin{align}\label{sd2}
\cM_{t}^{(k)}= \cM_{t-1}^{(k)}\cup \{s_t^{(k)}\}\cup_{\ell\in \cN_k} \cV_t^{(\ell)}, \quad \cM_{0}^{(k)}=\emptyset.
\end{align}
\end{enumerate}
In essence, each sensor stores new LLRs and relays them in the next time slot to its neighbours except that the newly collected sample is transmitted immediately at the same time slot. This is due to the setup that the sampling occurs before the message-passing within each time slot.   

Then the sample-dissemination-based distributed SPRT (SD-DSPRT) is performed at each sensor with the following stopping time and decision function:
\begin{align}\label{SHT_MI}
\T^{(k)}_\text{sd} \triangleq \min\left\{t: \zeta_t^{(k)}\triangleq \sum_{s\in \cM_t^{(k)}}\! \! s\;\notin \; (-A, B)\right\}, \quad D^{(k)}_\text{sd}\triangleq\left\{
\begin{array}{ll}
1, & \text{if}\;\;\sum_{s\in \cM_{\T^{(k)}_\text{sd}}^{(k)}}s\ge B,\\
0, &\text{if}\;\; \sum_{s\in \cM_{\T^{(k)}_\text{sd}}^{(k)}}s\le -A.
\end{array}\right.
\end{align}
Clearly, since the sample dissemination and the sequential test occur at the same time, $\cM_t^{(k)}\neq \{\{s_j^{(1)}\}_{j=1}^t, \{s_j^{(2)}\}_{j=1}^t, \ldots, \{s_j^{(K)}\}_{j=1}^t\}$ in general for $k=1, 2, \ldots, K$. In other words, the samples suffer from latency to reach all sensors in the network, which will potentially degrade the performance of $\T_\text{sd}^{(k)}$ compared to $\T_c$.  Note that the sample dissemination scheme under consideration may not provide the optimal routing strategy with respect to communication efficiency, but it guarantees that each sample is received by every sensor with least latency, which is beneficial in terms of minimizing the stopping time. In particular,
the information set at sensor $k$ and time $t$ is given by 
\begin{align}\label{claim1}
\cM_t^{(k)}=\left\{s^{(\ell)}_{\lb j-\nu_{\ell \to k}+1\rb^+}, \;\text{for}\; \; \ell=1, 2, \ldots, K\; \;\text{and}\; \; j = 1, 2, \ldots, t\right\},
\end{align} 
where  $\nu_{\ell \to k}$ is the length (number of links) of the shortest path from sensor $\ell$ to $k$, and $s^{(\ell)}_{0}\triangleq 0$ and $\nu_{k\to k}\triangleq 1$ for  notational convenience in the subsequent development.

The next result shows that the SD-DSPRT is order-$2$ asymptotically optimal. 
\begin{theorem}\label{thm0}
The asymptotic performance of the SD-DSPRT as $\alpha, \beta\to 0$ is characterized by
\begin{align}\label{MI_Opt}
\E_1\lb\T_\textnormal{sd}^{(k)}\rb\le \frac{-\log \alpha}{\sum_{k=1}^K{\cal D}_1^{(k)}}+O(1), \quad \E_0\lb\T_\textnormal{sd}^{(k)}\rb\le \frac{-\log \beta}{\sum_{k=1}^K{\cal D}_0^{(k)}}+O(1), \quad k=1, 2, \ldots, K.
\end{align}
\end{theorem}
\proof
On the account of the information set $\mathcal{M}_t^{(k)}$ in \eqref{claim1}, which is yielded by the sample dissemination process \eqref{sd1}-\eqref{sd2}, the decision statistic for SD-DSPRT at sensor $k$, i.e., the quantity $\zeta_t^{(k)}$ defined in \eqref{SHT_MI}, can be further written as
\begin{align}
\zeta_t^{(k)}=  \sum_{s\in \cM_t^{(k)}}\! \! s = \sum_{j=1}^t\sum_{\ell=1}^Ks^{(\ell)}_{\lb j-\nu_{\ell \to k}+1\rb^+}.
\end{align}
By noting that the stopping time at sensor $k$ is adapted to ${\cal M}_t^{(k)}$, i.e., the event $\{\T_\text{sd}^{(k)}\le t\}$ (or its complementary event $\{\T_\text{sd}^{(k)}> t\}$) is fully determined by $\cM_t^{(k)}$, we have
\begin{align}
&\E_i\lb \zeta_{\T_\text{sd}^{(k)}}^{(k)}\rb=\E_i\lb\sum_{j=1}^{\T_\text{sd}^{(k)}}\sum_{\ell=1}^Ks^{(\ell)}_{\lb j-\nu_{\ell \to k}+1\rb^+}\rb\nonumber\\=&\E_i\left[\sum_{j=1}^\infty\mathbbm{1}_{\{j\le {\T_\text{sd}^{(k)}}\}}\E_i\lb\left.\sum_{\ell=1}^Ks^{(\ell)}_{\lb j-\nu_{\ell \to k}+1\rb^+}\right|\cM_{j-1}^{(k)}\rb\right]\label{sd_eq1}\\=&\E_i\left[\sum_{j=1}^\infty\mathbbm{1}_{\{j\le {\T_\text{sd}^{(k)}}\}}\sum_{\ell=1}^K\E_i\lb\left.s^{(\ell)}_{\lb j-\nu_{\ell \to k}+1\rb^+}\right|\cM_{j-1}^{(k)}\rb\right]\nonumber\\=&\E_i\lb \sum_{j=1}^\infty\mathbbm{1}_{\{j\le {\T_\text{sd}^{(k)}}\}}\sum_{\ell=1}^K\underbrace{\E_i\lb s^{(\ell)}_{\lb j-\nu_{\ell \to k}+1\rb^+}\rb}_{\cD_i^{(\ell)}}\mathbbm{1}_{\{j\ge \nu_{\ell\to k}\}}\rb\label{sd_eq2}\\=&\E_i\left[\sum_{j=1}^\infty\left(\mathbbm{1}_{\{j\le \max_\ell\nu_{\ell\to k}\}}\sum_{\ell=1}^K\cD_i^{(\ell)}\mathbbm{1}_{\{j\ge \nu_{\ell\to k}\}}+\mathbbm{1}_{\{\max_\ell\nu_{\ell\to k}< j\le{\T_\text{sd}^{(k)}}\}}\sum_{\ell=1}^K\cD_i^{(\ell)}\mathbbm{1}_{\{j\ge \nu_{\ell\to k}\}}\right)\right]\label{sd_eq3}\\=&\E_i\left\{\sum_{j=1}^\infty\left[\mathbbm{1}_{\{j\le \max_\ell\nu_{\ell\to k}\}}\sum_{\ell=1}^K\cD_i^{(\ell)}\lb 1 -\mathbbm{1}_{\{j\le \nu_{\ell\to k}-1\}}\rb+\mathbbm{1}_{\{\max_\ell\nu_{\ell\to k}< j\le{\T_\text{sd}^{(k)}}\}}\sum_{\ell=1}^K\cD_i^{(\ell)}\right]\right\}\nonumber\\=&\E_i\left\{\sum_{j=1}^\infty\left[\underbrace{\lb\mathbbm{1}_{\{j\le \max_\ell\nu_{\ell\to k}\}}+\mathbbm{1}_{\{\max_\ell\nu_{\ell\to k}< j\le{\T_\text{sd}^{(k)}}\}}\rb}_{\mathbbm{1}_{\{ j\le \T_\text{sd}^{(k)}\}}}\sum_{\ell=1}^K\cD_i^{(\ell)}-\mathbbm{1}_{\{j\le \max_\ell\nu_{\ell\to k}\}}\sum_{\ell=1}^K\cD_i^{(\ell)}\mathbbm{1}_{\{j\le \nu_{\ell\to k}-1\}}\right]\right\}\nonumber\\=&\E_i\left[\sum_{j=1}^\infty\left(\mathbbm{1}_{\{ j\le \T_\text{sd}^{(k)}\}}\sum_{\ell=1}^K\cD_i^{(\ell)}-\sum_{\ell=1}^K\cD_i^{(\ell)}\mathbbm{1}_{\{j\le \nu_{\ell\to k}-1\}}\right)\right]\nonumber\\=&\E_i\lb\T_\text{sd}^{(k)}\sum_{\ell=1}^K\cD_i^{(\ell)}-\sum_{\ell=1}^K\cD_i^{(\ell)}\lb\nu_{\ell\to k}-1\rb\rb\nonumber\\=&\E_i\lb{\T_\text{sd}^{(k)}}\rb\sum_{\ell=1}^K\cD_i^{(\ell)}-\sum_{\ell=1}^K\lb\nu_{\ell\to k}-1\rb\cD_i^{(\ell)},\label{sd_eq4}
\end{align}
where \eqref{sd_eq1} holds due to Tower's property (i.e., $\E(X)= \E\left[\E\lb X|Y\rb\right]$) and the definition of the stopping time $\T_\text{sd}^{(k)}$; \eqref{sd_eq2} holds because $s^{(k)}_{0} = 0$ and $s^{(k)}_{\lb j-\nu_{\ell \to k}+1\rb^+}$ is independent of ${\cal M}_{j-1}^{(k)}$ due to \eqref{claim1}; \eqref{sd_eq3} is obtained by splitting $\mathbbm{1}_{\{j\le {\T_\text{sd}^{(k)}}\}}= \mathbbm{1}_{\{j\le \max_\ell\nu_{\ell\to k}\}} + \mathbbm{1}_{\{{\max_\ell\nu_{\ell\to k}<j\le \T_\text{sd}^{(k)}}\}}$. 

Under ${\cal H}_1$,  the local statistic $\zeta_{\T_\text{sd}^{(k)}}^{(k)}$ either hits the upper threshold (i.e., correct decision) with probability $1-\beta$ or the lower threshold (i.e. false alarm) with probability $\beta$. Thus its expected value upon stopping is expressed as 
\begin{align}\label{sd_eq5}
\E_1\lb \zeta_{\T_\text{sd}^{(k)}}^{(k)}\rb &=\beta \lb -A-\varsigma_0\rb+ (1-\beta) (B+\varsigma_1)\nonumber\\&\to B+O(1),\quad \text{as}\;\; A,B\to\infty,
\end{align}
where $\varsigma_i$'s  are the expected overshoots, which are constant terms (i.e., independent of $A, B$) that can be evaluated by renewal theory \cite{SeqA_book,Basseville}. Therefore, using \eqref{sd_eq4} and \eqref{sd_eq5}, we have
\begin{align}
&\E_1\lb\T_\text{sd}^{(k)}\rb=\frac{B}{\sum_{k=1}^K\cD_1^{(k)}}+\underbrace{\frac{\sum_{\ell=1}^K\lb \nu_{\ell\to k}-1\rb\cD_1^{(\ell)}+O(1)}{\sum_{k=1}^K\cD_1^{(k)}}}_{O(1)}.\label{MI_ET1}
\end{align}
Similarly, we can also obtain
\begin{align} 
\E_0\lb\T_\text{sd}^{(k)}\rb=\frac{A}{\sum_{k=1}^K\cD_0^{(k)}}+\underbrace{\frac{\sum_{\ell=1}^K\lb \nu_{\ell\to k}-1\rb\cD_0^{(\ell)}+O(1)}{\sum_{k=1}^K\cD_0^{(k)}}}_{O(1)}.\label{MI_ET0}
\end{align}

On the other hand, since $\zeta_{\T_\text{sd}^{(k)}}^{(k)}$ is the sum of independent LLRs, it is readily obtained by the Markov inequality that
\begin{align}
&\alpha\triangleq\Prob_0\lb \zeta_{\T_\text{sd}^{(k)}}^{(k)}\ge B\rb\le e^{-B}\,\E_0\left[ \exp\lb{\zeta_{\T_\text{sd}^{(k)}}^{(k)}}\rb\right]=e^{-B},\label{MI_Error0}\\
&\beta\triangleq \Prob_1\lb \zeta_{\T_\text{sd}^{(k)}}^{(k)}\le -A\rb\le e^{-A}\,\E_1\left[ \exp\lb{-\zeta_{\T_\text{sd}^{(k)}}^{(k)}}\rb\right]=e^{-A}.\label{MI_Error1}
\end{align}
The equalities in \eqref{MI_Error0} and \eqref{MI_Error1} follow from the optional sampling theorem \cite{Lawler} by noting that $\exp\lb{\zeta_{\T_\text{sd}^{(k)}}^{(k)}}\rb$ and $\exp\lb{-\zeta_{\T_\text{sd}^{(k)}}^{(k)}}\rb$ are martingales under ${\cal H}_0$  and ${\cal H}_1$ respectively. In specific, $\E_0\left[\exp\lb{\zeta_{\T_\text{sd}^{(k)}}^{(k)}}\rb\right]=\E_0\left[\exp\lb{\zeta_{0}^{(k)}}\right] \rb=1$, and $\E_1\left[\exp\lb{-\zeta_{\T_\text{sd}^{(k)}}^{(k)}}\rb\right]=\E_1\left[\exp\lb{-\zeta_{0}^{(k)}}\rb \right]=1$.

Combining \eqref{MI_ET1}-\eqref{MI_Error1} leads to the results in \eqref{MI_Opt}.
\endproof
\begin{remark}
According to \eqref{MI_ET0} and \eqref{MI_ET1} in the proof of Theorem \ref{thm0}, the condition that every sample reaches all sensors via the shortest paths is sufficient but not necessary for the order-$2$ asymptotic optimality. 
In particular, we can further relax $\nu_{\ell\to k}$ in \eqref{MI_ET0} and \eqref{MI_ET1} to be any finite number (i.e., samples travel from sensor $\ell$ to $k$ within finite number of hops), and still preserve the constant terms, which are essential for  the order-$2$ optimality. However, the resulting scheme yields larger constant deviation from the centralized test than that in the proposed scheme, thus is less efficient in terms of the stopping time.
\end{remark}
Note that the bounds in \eqref{MI_Error0} and \eqref{MI_Error1} provide accurate characterizations for the error probabilities, as shown in Section V. Therefore, in practice, the sequential thresholds can be set according to $A=-\log \beta$ and $B=-\log \alpha$.

\ignore{At the first glance, such a scheme requires increasing storage for the growing number of statistics throughout the network since ${\cal M}_t^{(k)}$ keeps growing with $t$. However, note that at time $t$, sensor $k$ has received LLRs generated at time $j<t- \max_\ell\nu_{\ell\to k}$ from all other sensors, and there is no point of distinguish them from one another since the decision statistic pertains to the sum of LLRs. In other words, it is only necessary to store LLRs from different sensors for $j= t, t-1, \ldots, t-\max_\ell\nu_{\ell\to k}$, the sum of all LLRs before $t-\max_\ell\nu_{\ell\to k}$, i.e., the maximum number of LLRs that is stored is $K\max_\ell\nu_{\ell\to k}$. 
For fixed network topology, $\max_\ell\nu_{\ell\to k}$ is a constant, thus a constant memory is needed to implement the SD-DSPRT algorithm.  }

Although the distributed sequential test SD-DSPRT achieves the order-$2$ asymptotically optimal performance at every sensor, it is at the cost of the significant communication overhead that arises from the exchange of sample arrays with the additional index information. In particular, an increase in the network size $K$ will significantly increase the dimension of sample array and the index information, making the sample dissemination practically infeasible. In the next section, we consider another message-passing based distributed sequential test that avoids the high communication overhead, yet still achieves the same order-$2$ asymptotic optimality at all sensors.

\section{Asymptotic Optimality of Distributed Sequential Test via Consensus Algorithm}
In this section, we consider the distributed sequential test based on the communication protocol known as  the consensus algorithm, in which the sensors exchange their local decision statistics instead of the raw samples (which is an array of messages), i.e., $\cV_t^{(k)}$ only contains a scalar. Moreover, we assume that $q$ rounds of message-passings can take place within each sampling interval. Denoting the decision statistic at sensor $k$ and time $t$ as $\eta_t^{(k)}$, then during every time slot $t$, the consensus-algorithm-based sequential test is carried out as follows:
\begin{enumerate}
\item Take a new sample, and add the LLR $s_t^{(k)}$ to the local decision statistic from previous time:
\begin{align}\label{protocol_1}
\widetilde \eta^{(k)}_{t, 0}=\eta^{(k)}_{t-1}+s_t^{(k)},
\end{align}
where $\widetilde \eta^{(k)}_{t, 0}$ is the intermediate statistic before message-passing, and we denote the statistic after $m$th message-passing as $\widetilde \eta^{(k)}_{t, m}, \; m= 0, 1, 2, \ldots, q$ which is computed in the next step.
\item For $m = 0, 1, 2, \ldots, q$, every sensor exchanges its local intermediate statistic $\widetilde \eta^{(k)}_{t, m}$ with the neighbours, and updates the local intermediate statistic as the weighted sum of the available statistics from the neighbours, i.e., 
\begin{align}\label{protocol_2}
\widetilde\eta^{(k)}_{t, m}=w_{k,k}\,\widetilde\eta^{(k)}_{t, m-1}+\sum_{\ell\in {\cal N}_k}w_{\ell, k}\, \widetilde \eta^{(\ell)}_{t, m-1},\;\; \text{for}\;\; m=1, 2, \ldots, q,
\end{align}
where the weight coefficients $w_{i,j}$ will be specified later.
\item Update the local decision statistic for time $t$ as $\eta^{(k)}_{t}=\widetilde\eta^{(k)}_{t, q}$.
\item Go to Step 1) for the next sampling time slot $t+1$.
\end{enumerate}

To express the consensus algorithm in a compact form, we define the following vectors: $$\widetilde\bEta_{t, m}\triangleq [\widetilde\eta^{(1)}_{t, m}, \widetilde\eta^{(2)}_{t, m}, \ldots, \widetilde\eta^{(K)}_{t, m}]^T, \;\; \bEta_t\triangleq [\eta^{(1)}_t, \eta^{(2)}_t, \ldots, \eta^{(K)}_t]^T,$$ $$\bs_t\triangleq [s^{(1)}_t, s^{(2)}_t, \ldots, s^{(K)}_t]^T.$$ Then each message-passing in \eqref{protocol_2} can be represented by\begin{align}
\widetilde\bEta^{(k)}_{t, m}=\bW\widetilde\bEta^{(k)}_{t, m-1}, \;\;\text{for}\;\; m=1, 2, \ldots, q,
\end{align}
where the matrix $\bW\triangleq (w_{i,j})\in \mathbb{R}^{K\times K}$ is formed by $w_{i,j}$'s defined in \eqref{protocol_2}. Combining \eqref{protocol_1} and \eqref{protocol_2}, the decision statistic vector  evolves over time according to 
\begin{align}\label{protocol_sys}
  \bEta_{t}=\bW^q\lb\bEta_{t-1}+\bs_{t}\rb, \quad \text{with}\;\; \bEta_0={\bf 0}.
\end{align}
Based on \eqref{protocol_sys}, the decision statistic vector at time $t$ can also be equivalently expressed as
\begin{align}\label{G_stats}
  \bEta_{t}=\sum_{j=1}^t\bW^{q\lb t-j+1\rb}\bs_j, \quad t=1, 2, \ldots.
\end{align}

As such, the consensus-algorithm-based distributed SPRT (CA-DSPRT) at sensor $k$ can be  implemented with the following stopping time and decision rule:
\begin{align}
  \T_\text{ca}^{(k)}\triangleq\inf\left\{t: \eta_t^{(k)}\notin (-A,B)\right\},\quad D^{(k)}_\text{ca}\triangleq \left\{
\begin{array}{ll}
1 & \text{if}\quad \eta^{(k)}_{\T_\text{ca}^{(k)}}\ge B,\\
0 & \text{if}\quad \eta^{(k)}_{\T_\text{ca}^{(k)}}\le -A,
\end{array}\right.
\end{align}
where $\{A, B\}$ are chosen to satisfy the error probability constraints.

Note that \eqref{protocol_sys} resembles the consensus algorithm in the {\it fixed-sample-size test} \cite{Braca10}, where no innovation are introduced, i.e., $ \bEta_{t}=\bW^q\bEta_{t-1}$. In that case, under certain regularity conditions for $\bW$, consensus is reached in the sense $\bEta_t\to \left[\frac{1}{K}\sum_{i=1}^K\eta^{(k)}_0, \ldots, \frac{1}{K}\sum_{i=1}^K\eta^{(k)}_0\right]^T$ as $t\to \infty$. In contrast, with the new samples constantly arriving, how such a message-passing protocol can affect the {\it sequential test} at each sensor has not been investigated in the literature. In the following subsection, we will show that the above CA-DSPRT enables every sensor to attain the order-$2$ asymptotically optimal test performance, instead of reaching consensus on teh decision statistics. 


\subsection{Order-$2$ Asymptotic Optimality of CA-DSPRT}
To begin with, we first impose the following two conditions on the weight matrix $\bW$ and the distribution of LLR respectively.
\begin{condition}\label{con1}
The weight matrix $\bW$ satisfies $$\bW{\bf 1}={\bf 1},\;\; {\bf 1}^T\bW={\bf 1}^T, \;\; 0< \sigma_\text{2}\lb \bW\rb <1,$$
where $\sigma_{i}\lb\bW\rb$ denotes $i$th singular value of $\bW$.
\end{condition}

\begin{condition}\label{con2}
The LLR for the hypothesis testing problem satisfies that $\E_i\lb e^{K\sqrt{K}|s_j^{(k)}|}\rb$ is bounded for $ i\in\{0, 1\}, \; k=1,  \ldots, K$.
\end{condition}
The first condition essentially regulates the network topology and weight coefficients in \eqref{protocol_sys}. If we further require $w_{i,j}\ge 0$, then Condition \ref{con1} is equivalent to $\bW$ being doubly stochastic. The second condition regulates the tail distribution of the LLR at each sensor, which in fact embraces a wide range of distributions, for example, the Gaussian and Laplacian distributions. 

\begin{theorem}\label{thm1}
Given that Conditions \ref{con1}-\ref{con2} are satisfied, the asymptotic performance of the CA-DSPRT as $\alpha, \beta\to 0$ is characterized by
\begin{align}
\E_1\lb\T_\text{ca}^{(k)}\rb\le \frac{-\log \alpha}{\sum_{k=1}^K{\cal D}_1^{(k)}}+\frac{\sigma_2^q(\bW)}{1-\sigma_2^q(\bW)}O(1), \quad \E_0\lb\T_\text{ca}^{(k)}\rb\le \frac{-\log \beta}{\sum_{k=1}^K{\cal D}_0^{(k)}}+\frac{\sigma_2^q(\bW)}{1-\sigma_2^q(\bW)}O(1).
\end{align}
Therefore, the CA-DSPRT achieves the order-$2$ asymptotically optimal solution to \eqref{ST} for $k=1, 2, \ldots, K$.
\end{theorem}

Theorem \ref{thm1} can be readily proved by invoking the following two key lemmas.
\begin{lemma}\label{lemma1}
All sensors achieve the same expected stopping time in the asymptotic regime:
\begin{align}\label{Lemma1_eq}
\E_1\lb \T_\text{ca}^{(k)}\rb=\frac{B}{\sum_{k=1}^K\cD_1^{(k)}/K}+O(1), \quad \E_0\lb \T_\text{ca}^{(k)}\rb=\frac{A}{\sum_{k=1}^K\cD_0^{(k)}/K}+O(1), 
\end{align}
for $k=1, 2, \ldots, K$, as $A, B\to \infty$.
\end{lemma}
\proof
For notational convenience, we omit the subscript of $\T_\text{ca}^{(k)}$ and use $\T^{(k)}$ for the stopping time of the CA-DSPRT throughout the proof. 

We first define $\bJ\triangleq \frac{1}{K}{\bf 1}{\bf 1}^T$, where ${\bf 1}$ is an all-one  vector. Note that the following equality will become useful in our proof later:
\begin{align}\label{Wt}
\bW^{t}-\bJ=\lb\bW-\bJ\rb^t, \quad \text{for}\;\; t=1, 2, \ldots,
\end{align}
which can be shown by induction as follows: 1) For $t=1$, \eqref{Wt} obviously holds true; 2) assume $\bW^n-\bJ=\lb\bW-\bJ\rb^n$, then 
\begin{align}
\lb\bW-\bJ\rb^{n+1}&=\lb\bW-\bJ\rb^{n}\lb\bW-\bJ\rb\nonumber\\&=\lb\bW^n-\bJ\rb\lb\bW-\bJ\rb\nonumber\\&=\bW^{n+1}-\bJ\bW-\bW^n\bJ+\bJ^2\nonumber\\&=\bW^{n+1}-\bJ,
\end{align}
where the last equality holds true because Condition \ref{con1} implies that $\bJ\bW=\frac{1}{K}{\bf 1}{\bf 1}^T\bW=\frac{1}{K}{\bf 1}{\bf 1}^T=\bJ$,  and furthermore $$\bW^n\bJ=\bW^{n-1}\lb\frac{1}{K}\bW{\bf 1}{\bf 1}^T\rb=\bW^{n-1}\bJ=\cdots =\bJ,$$ and $\bJ\bJ=\frac{1}{K^2}{\bf 1}{\bf 1}^T{\bf 1}{\bf 1}^T=\bJ$ follows by definition.

Another useful inequality holds for any matrix $\boldsymbol\Theta\in \mathbb{R}^{L\times L}$ and $\bx\in \mathbb{R}^L$ \cite{Matrix_book}
\begin{align}\label{bound_eig}
\frac{\lVert\boldsymbol\Theta\bx\rVert_2}{\lVert\bx\rVert_2}\le \sup_{\bx\in \mathbb{R}^L}\frac{\lVert\boldsymbol\Theta\bx\rVert_2}{\lVert\bx\rVert_2}=\sigma_1\lb\boldsymbol\Theta\rb,\end{align}
where $\lVert\cdot\rVert_2$ is the $L_2$-norm, and $\sigma_1(\cdot)$ is the largest singular value of a given matrix.
Moreover, Condition \ref{con1} implies that $\bW$ has the maximum singular value $\sigma_1(\bW)=1$, and $\bW=\frac{1}{K}{\bf 1}{\bf 1}^T+\sum_{i=2}^K\sigma_i(\bW)\bu_i\bv_i^T$, where $\bu_i$ and $\bv_i$ are singular vectors associated with $\sigma_i\lb\bW\rb$, leading to
\begin{align}
\sigma_1 \lb\bW-\bJ\rb=\sigma_2(\bW).
\end{align} 
For notational simplicity, $\sigma_2$ will represent $\sigma_2(\bW)$ henceforth unless otherwise stated. Substituting $\boldsymbol\Theta=\bW-\bJ$ into \eqref{bound_eig},  we have the following bounds for any random vector $\bs_j$ (that consists of LLRs at time $j$):
\begin{align}\label{boundstats}
\lVert \lb\bW-\bJ\rb^{q(t-j+1)} \bs_j\rVert_2&= \lVert \lb\bW-\bJ\rb \lb\bW-\bJ\rb^{q(t-j+1)-1}\bs_j\rVert_2 \nonumber\\&\le\sigma_2 \lVert \lb\bW-\bJ\rb^{q(t-j+1)-1}\bs_j\rVert_2\nonumber\\&\le \sigma_2^2 \lVert \lb\bW-\bJ\rb^{q(t-j+1)-2}\bs_j\rVert_2\nonumber\\&\cdots\nonumber\\&\le\sigma_2^{q(t-j+1)}\lVert\bs_j\rVert_2.
\end{align}
Denoting $\be_k\triangleq[0, \ldots, \underbrace{1}_{k\text{th element}}, \ldots, 0]^T$ and invoking \eqref{Wt} and \eqref{boundstats} give the following inequalities 
\begin{align}\label{boundproof_eq1}
\left|\be_k^T\lb\bW^{q(t-j+1)}-\bJ\rb\bs_j\right|\le \lVert \lb\bW^{q(t-j+1)}-\bJ\rb \bs_j\rVert_2\le \sigma_2^{q(t-j+1)}\lVert \bs_j\rVert_2 ,\;\;\text{a.s.}.
\end{align}
Then expanding the leftmost term in \eqref{boundproof_eq1} gives
\begin{align}\label{bound_intermidiate}
-\sigma_2^{q(t-j+1)}\lVert\bs_j\rVert_2 +\be_k^T \bJ\bs_j\le\be_k^T\bW^{q(t-j+1)}\bs_j\le \sigma_2^{q(t-j+1)}\lVert \bs_j\rVert_2+\be_k^T \bJ\bs_j,\;\; \text{a.s.}.
\end{align}
Summing \eqref{bound_intermidiate} from $j=1$ to $j=t$, and using \eqref{G_stats}, we have
\begin{align}\label{bound_local}
-\sum_{j=1}^t\lVert\bs_j\rVert_2 \sigma_2^{q(t-j+1)}+\be_k^T \bJ\sum_{j=1}^t\bs_j\le\be_k^T&\underbrace{\sum_{j=1}^t\bW^{q(t-j+1)}\bs_j}_{\boldsymbol\eta_t}\nonumber\\&\le\sum_{j=1}^t\lVert \bs_j\rVert_2 \sigma_2^{q(t-j+1)}+\be_k^T \bJ\sum_{j=1}^t\bs_j,\;\;\text{a.s.}
\end{align}
for any $t=1, 2, \ldots$. 
Taking expectations on both inequalities of \eqref{bound_local}, we arrive at 
\begin{align}\label{bound_exp}
-\E_i\lb\sum_{j=1}^{\T^{(k)}}\lVert\bs_j\rVert_2 \sigma_2^{q(\T^{(k)}-j+1)}\rb&+\be_k^T \bJ\E_i\lb\sum_{j=1}^{\T^{(k)}}\bs_j\rb\le\E_i\lb\be_k^T\bEta_{\T^{(k)}}\rb=\E_i\lb\eta^{(k)}_{\T^{(k)}}\rb\nonumber\\&\le\E_i\lb\sum_{j=1}^{\T^{(k)}}\lVert \bs_j\rVert_2 \sigma_2^{q(\T^{(k)}-j+1)}\rb+\E_i\lb\be_k^T \bJ\sum_{j=1}^{\T^{(k)}}\bs_j\rb\!, \;\; i=0, 1.
\end{align}
Let us look at the first inequality in \eqref{bound_exp} first. We have
\begin{align}\label{BoundJS}
\be_k^T \bJ\E_i\lb\sum_{j=1}^{\T^{(k)}}\bs_j\rb&\le\E_i\lb\eta^{(k)}_{\T^{(k)}}\rb+\E_i\lb\sum_{j=1}^{\T^{(k)}}\lVert\bs_j\rVert_2 \sigma_2^{q(\T^{(k)}-j+1)}\rb,
\end{align}
where the second term on the right-hand side can be further bounded above by
\begin{align}
\E_i\lb\sum_{j=1}^{\T^{(k)}}\lVert\bs_j\rVert_2 \sigma_2^{q(\T^{(\ell)}-j+1)}\rb&\le \E_i\lb\sup_{t}\sum_{j=1}^t\lVert\bs_j\rVert_2\sigma_2^{q(t-j+1)}\rb\nonumber\\&=\E_i\lb\sup_{t}\sum_{j=1}^t\lVert\bs_j\rVert_2\sigma_2^{qj}\rb\label{BoundConst1}\\&=\E_i\lb\sum_{j=1}^\infty\lVert\bs_j\rVert_2\sigma_2^{qj}\rb=\frac{\sigma_2^q}{1-\sigma_2^q}\,\E_i\lVert\bs_j\rVert_2,\label{BoundConst}
\end{align}
where \eqref{BoundConst1} holds since $\bs_j$ are independent and identically distributed  for all $j$.

Meanwhile, the left-hand side of \eqref{BoundJS} for $i=1$ (i.e., under ${\cal H}_1$) can be expressed as
\begin{align}
\be_k^T \bJ\,\E_1\lb\sum_{j=1}^{\T^{(k)}}\bs_j\rb&=\be_k^T \frac{1}{K}{\bf 1}{\bf 1}^T\,\E_1\lb\sum_{j=1}^{\infty}\mathbbm{1}_{\{j\le \T^{(k)}\}}\bs_j\rb\nonumber\\&=\be_k^T \frac{1}{K}{\bf 1}{\bf 1}^T\,\E_1\lb\sum_{j=1}^{\infty}\mathbbm{1}_{\{j\le \T^{(k)}\}}\E_1\lb\left.\bs_j\right|\bs_{j-1}, \bs_{j-2}, \ldots, \bs_1\rb\rb\label{towerProp}\\&=\be_k^T \frac{1}{K}{\bf 1}{\bf 1}^T\,\E_1\lb\sum_{j=1}^{\infty}\mathbbm{1}_{\{j\le \T^{(k)}\}}\underbrace{\E_1\lb\bs_j\rb}_{[\cD_1^{(1)}, \cD_1^{(2)}, \ldots, \cD_1^{(K)}]^T}\rb\nonumber\\&=\be_k^T \frac{1}{K}{\bf 1}\underbrace{{\bf 1}^T\,{[\cD_1^{(1)}, \cD_1^{(2)}, \ldots, \cD_1^{(K)}]^T}}_{\sum_{k=1}^K\cD_1^{(k)}}\,\E_1\lb\T^{(k)}\rb\nonumber\\&=\E_1\lb\T^{(k)}\rb\sum_{k=1}^K\cD_1^{(k)}/K,\label{vWald}
\end{align}
where \eqref{towerProp} is obtained by the Tower's property and the fact that $\{\T^{(k)}\ge j\}$ (or its complementary event $\{\T^{(k)}\le  j-1 \}$) is fully determined by  $\bs_{j-1}, \bs_{j-2}, \ldots, \bs_1$. 

Combining \eqref{BoundJS}, \eqref{BoundConst}, and \eqref{vWald} for $i=1$ gives
\begin{align}\label{Bound_Ineq}
\E_1\lb\T^{(k)}\rb\le \frac{\E_1\lb\eta^{(k)}_{\T^{(k)}}\rb}{\sum_{k=1}^K{\cD}^{(k)}_1/K}+\frac{\sigma_2^q}{1-\sigma_2^q}\frac{\E_1\lb\lVert\bs_j\rVert_2\rb}{\sum_{k=1}^K{\cD}^{(k)}_1/K}.
\end{align}
Note that, under ${\cal H}_1$, $\eta^{(k)}_{\T^{(k)}}$ either hits the upper threshold with probability $1-\beta$ or the lower threshold with probability $\beta$, i.e., 
\begin{align}
\E_1\lb\eta^{(k)}_{\T^{(k)}}\rb&=\beta \lb -A-\varsigma_0\rb+ (1-\beta) (B+\varsigma_1)\nonumber\\&\to B+O(1),\quad A,B\to\infty,
\end{align}
with the constant expected overshoots $\varsigma_i$'s (i.e., independent of $A, B$) that can be evaluated by renewal theory \cite{SeqA_book,Basseville}. 

Moreover, by noting that $\sqrt{K}|s_j^{(k)}|<1+K\sqrt{K}|s_j^{(k)}|\le e^{K\sqrt{K}|s_j^{(k)}|}$ (since $1+x\le e^x$), then Condition \ref{con2} indicates that
\begin{align}
\sqrt{K}\,\E_i\lb |s_j^{(k)}|\rb< \E_i\lb e^{K\sqrt{K}|s_j^{(k)}|}\rb\le C, \quad k=1, 2, \ldots, K, 
\end{align}
which, together with the relation between the $L_2$ and $L_\infty$ norms, further implies
\begin{align}
\E_i\lb\lVert\bs_j\rVert_2\rb\le \sqrt{K}\,\E_i\lb\lVert\bs_j\rVert_\infty\rb \triangleq\sqrt{K} \max_k \E_i\lb |s_j^{(k)}|\rb< C.
\end{align}
As a result, Condition \ref{con2} provides the sufficient condition such that $\E_i\lb\lVert\bs_j\rVert_2\rb$ is bounded above by some constant, and hence $\E_i\lb\lVert\bs_j\rVert_2\rb= O(1)$. 

Therefore, the following inequality follows from \eqref{Bound_Ineq}:
\begin{align}\label{UpperBound}
\E_1\lb\T^{(k)}\rb\le \frac{B}{\sum_{k=1}^K{\cD}^{(k)}_1/K}+\frac{\sigma_2^{q}}{1-\sigma_2^q}O(1), \qquad A, B\to \infty.
\end{align}
Similarly, from the second inequality in \eqref{bound_exp}, we can establish 
\begin{align}
\E_1\lb\T^{(k)}\rb\ge \frac{B}{\sum_{k=1}^K{\cD}^{(k)}_1/K}-\frac{\sigma_2^{q}}{1-\sigma_2^q}O(1), \qquad A, B\to \infty, 
\end{align}
which, together with \eqref{UpperBound}, proves the asymptotic characterization for $\E_1\lb\T^{(k)}\rb$ given by \eqref{Lemma1_eq}.

By treading on the similar derivations as above, $\E_0\lb\T^{(k)}\rb$ can be bounded by
\begin{align}
\frac{A}{\sum_{k=1}^K{\cD}^{(k)}_1/K}-\frac{\sigma_2^{q}}{1-\sigma_2^q}O(1)\le 
\E_0\lb\T^{(k)}\rb\le \frac{A}{\sum_{k=1}^K{\cD}^{(k)}_0/K}+\frac{\sigma_2^{q}}{1-\sigma_2^q}O(1), \quad A, B\to \infty,
\end{align}
which completes the proof.
\endproof
Lemma \ref{lemma1} characterizes how the expected sample sizes of the CA-DSPRT vary as the decision thresholds go to infinity. 
The next lemma relates the error probabilities of the CA-DSPRT in the same asymptotic regime to the decision thresholds. 
\begin{lemma}\label{lemma2}
The error probabilities of CA-DSPRT in the asymptotic regime as $A, B\to\infty$ at each sensor is bounded above by
\begin{align}\label{lemma2_eq}
\log \Prob_0\lb D^{(k)}_\text{ca}=1\rb\le  -KB+O(1), \quad \log \Prob_1\lb D^{(k)}_\text{ca}=0\rb\le -KA+O(1).
\end{align}
\end{lemma}
\begin{proof}
Again, the proof makes use of the inequality \eqref{bound_local} to bound the local statistic. In the following, we show the proof for the Type-I error probability, while that for the Type-II error  probability follows similarly. 

First, due to \eqref{bound_local}, note the following relation $$\left\{D_\text{ca}^{(k)}=1\right\}\triangleq \left\{\be_k^T\sum_{j=1}^{\T_\text{ca}^{(k)}}\bW^{q\lb\T_\text{ca}^{(k)}-j+1\rb}\bs_j\ge B\right\}\subset \left\{\sum_{j=1}^{\T_\text{ca}^{(k)}}\lVert \bs_j\rVert_2 \sigma_2^{q\lb\T_\text{ca}^{(k)}-j+1\rb}+\be_k^T \bJ\sum_{j=1}^{\T_\text{ca}^{(k)}}\bs_j\ge B\right\}.$$ 
Therefore,
\begin{align}\label{const_proof}
\Prob_0\lb D^{(k)}_\text{ca}=1\rb&\le \Prob_0\lb \underbrace{\sum_{j=1}^{\T_\text{ca}^{(k)}}\lVert \bs_j\rVert_2 \sigma_2^{q\lb\T_\text{ca}^{(k)}-j+1\rb}}_{\phi_{\T_\text{ca}^{(k)}}}+\be_k^T \bJ\sum_{j=1}^{\T_\text{ca}^{(k)}}\bs_j\ge B\rb\nonumber\\&= \Prob_0\lb \exp\left[K\lb{\phi_{\T_\text{ca}^{(k)}}}+\be_k^T \bJ\sum_{j=1}^{\T_\text{ca}^{(k)}}\bs_j\rb\right]\ge e^{KB}\rb\nonumber\\&\le e^{-KB} \;\underbrace{\E_0\lb \exp\left[K\lb \phi_{\T_\text{ca}^{(k)}}+\be_k^T \bJ\sum_{j=1}^{\T_\text{ca}^{(k)}}\bs_j\rb\right]\rb}_{{\cal B}_k}
\end{align}
where the second inequality follows from the Markov inequality. 

In order to show the results in \eqref{lemma2_eq},  the remaining task is to bound the coefficient term
\begin{align}\label{const_proof2}
{\cal B}_k&=  \;\E_0\lb e^{K\phi_{\T_\text{ca}^{(k)}}}\exp\left[\sum_{\ell=1}^K \sum_{j=1}^{\T_\text{ca}^{(k)}}s^{(\ell)}_j\right]\rb 
\nonumber\\&=\;\E_0\lb e^{K\phi_{\T_\text{ca}^{(k)}}} \prod_{j=1}^{\T_\text{ca}^{(k)}}\prod_{\ell=1}^Kl_j^{(\ell)}\rb \nonumber\\&=\;\E_1\lb e^{K\phi_{\T_\text{ca}^{(k)}}} \rb,
\end{align}
where the last equality is obtained by changing the probability measure of the expectation from ${\cal H}_0$ to ${\cal H}_1$. 
To that end, the following inequalities are useful
\begin{align}\label{const_bound}
\;\E_1\lb e^{K\phi_{\T_\text{ca}^{(k)}}} \rb\le \;\E_1\lb e^{K \sup_t \phi_t} \rb= \;\prod_{j=1}^\infty\E_1\lb  e^{K \lVert\bs_j\rVert_2\sigma_2^{qj}} \rb&\le \;\prod_{j=1}^\infty\lb\E_1\lb  e^{K \lVert\bs_j\rVert_2} \rb\rb^{\sigma_2^{qj}},
\end{align}
where the second inequality follows from the Jenson's inequality since $x^a$ is a concave function for $a<1$. Thanks to \eqref{const_bound}, Condition \ref{con2} (i.e., there exists a finite number $M$ such that $\E_i\lb  e^{K \sqrt{K}\lVert\bs_j\rVert_\infty} \rb\le M$) is sufficient to ensure that ${\cal B}_k$ in \eqref{const_proof} is upper bounded by a constant term (i.e., independent of $A, B$) due to the following: 
\begin{align}
{\cal B}_k\le \prod_{j=1}^\infty\lb\E_1\lb  e^{K \lVert\bs_j\rVert_2} \rb\rb^{\sigma_2^{qj}}&\le   \prod_{j=1}^\infty\lb\underbrace{\E_1\lb  e^{K \sqrt{K}\lVert\bs_j\rVert_\infty} \rb}_{\le M}\rb^{\sigma_2^{qj}}\nonumber\\&\le \exp\lb\sum_{j=1}^\infty\sigma_2^{qj}\log M\rb=M^{\frac{\sigma_2^q}{1-\sigma_2^q}} = O(1).
\end{align}
As a result, \eqref{const_proof} implies that
\begin{align}
\log \Prob_0\lb D^{(k)}_\text{ca}\ge B\rb\le -KB+\frac{\sigma_2^q}{1-\sigma_2^q}O(1),
\end{align}
proving the asymptotic characterization of the Type-I error probability given by \eqref{lemma2_eq}. 
\end{proof}

\ignore{\subsection{Large Number of Sensors}
In practice, the network size can be large, especially in the case where the target signal is weak, and a large number of sensors are desired in the hope of an acceptable stopping time. This scenario encourages the asymptotic analysis of CA-DSPRT in the regime where $K\to\infty$. Accordingly, it is necessary to  assume the total K-L divergence in the network amounts to a constant as $K$ grows; otherwise, sequential testing becomes not meaningful, since the fixed-sample test with one sampling step would lead to vanishing error probabilities as $K\to \infty$. Therefore, we denote 
\begin{align}
&{\cD}_i^{(k)}= \widetilde \cD_i^{(k)}/K+o\lb\frac{1}{K}\rb, \nonumber\\\text{and}
\;\;&{\sum_{k=1}^K\cD_i^{(k)}}\to\overline{\cD}_i \triangleq\sum_{k=1}^K{\widetilde  \cD}_i^{(k)}/K\sim O(1), \quad \text{as}\;\; K\to \infty,\;\; i=0, 1.
\end{align}
Note that $\overline{D}_i$ is a constant term in terms of growing $K$, indicating that the summing KL divergence of the entire network does not scale with the network size.
This corresponds to the case where the practitioner set a desired performance, and then decide the network size depending on the strength of the target signal. And as the signal strength becomes small and the network size grows large, what performance tradeoff can be achieved in the asymptotic regime. 

Another assumption we need when scaling up the network is that $\sigma_2^q$ needs to be fixed as the network size grows. This is accomplished by either fixing $\sigma_2$, i.e., confining the network topology in the class with the same $\sigma_2$, or  increasing $q$, i.e., the number of  message-passing at each step. The practitioners can design the value of $q$ depending on the value of changing $\sigma_2$. 
\begin{condition}\label{con3}
$\E_i\lb s_j^{(k)}\rb=O\lb\frac{1}{K}\rb$ and $\E_i\lb e^{K\sqrt{K}s_j^{(k)}}\rb\le M^{\sqrt{K}}$.
\end{condition}
\begin{theorem}\label{thm2}
Given that the Condition \ref{con1} on the message-passing matrix, and the Condition \ref{con2} on the tail probability of the statistic are satisfied, then the Con-SPRT possesses the order-1 asymptotic optimality for $A, B\to \infty$ and $K\to \infty$:
\begin{align}
&\E_1\lb\T^{(k)}\rb\sim \frac{-\log \alpha}{\sum_{k=1}^K{\cal D}_1^{(k)}}+o({-\log \alpha}), \\ &\E_0\lb\T^{(k)}\rb\sim \frac{-\log \beta}{\sum_{k=1}^K{\cal D}_0^{(k)}}+o({-\log \beta}).
\end{align}
\end{theorem}
The growing rate of $A,B$ and $K$ can be arbitrary. 
\proof
Recalling the inequality in \eqref{Bound_Ineq} and $\lVert\bs_j\rVert_2\le \sqrt{K}\lVert\bs_j\rVert_\infty$, we have
\begin{align}
\E\T\sum_{k=1}^K\cD^{(k)}/K\le \E\lb\eta^{(k)}_{\T^{(k)}}\rb+\frac{\sigma_2^q}{1-\sigma_2^q}\sqrt{K}\,\E\lb \lVert \bs_j\rVert_\infty\rb.
\end{align}
\begin{align}
\E\T&\le \frac{K\E\lb\eta^{(k)}_{\T^{(k)}}\rb}{\sum_{k=1}^K\cD^{(k)}}+\frac{\sigma_2^q}{1-\sigma_2^q}\frac{K \sqrt{K}\,\E\lb \lVert \bs_j\rVert_\infty\rb}{\sum_{k=1}^K\cD^{(k)}}\nonumber\\&\le \frac{KB}{\sum_{k=1}^K\cD^{(k)}}+\frac{\sigma_2^q}{1-\sigma_2^q}{O(\sqrt{K})}.
\end{align}
\begin{align}
\Prob_0\lb D_\T=1\rb&\le e^{-KB}\prod_{j=1}^\infty\Big(\underbrace{\E_1\lb  e^{K\sqrt{K} \lVert\bs_j\rVert_\infty} \rb}_{M^{\sqrt{K}}}\Big)^{\sigma_2^{qj}}\nonumber\\&=\exp\lb-KB+\sqrt{K}\sum_{j=1}^\infty {\sigma_2^{qj}} \log M\rb\nonumber\\&=\exp\lb-KB+\sqrt{K}\frac{\sigma_2^q}{1-\sigma_2^q} \log M\rb
\end{align}
\begin{align}
\E\T &\le \frac{-\log \alpha}{\sum_{k=1}^K{\cal D}_1^{(k)}}+O\lb{\sqrt{K}}\rb, \nonumber\\&=\frac{-\log \alpha}{\sum_{k=1}^K{\cal D}_1^{(k)}}+o\lb\log\alpha\rb\quad \text{as}\;\; K\to \infty \;(\text{thus}\;\; \alpha\to 0).
\end{align}
Thus 
\begin{align}
1\le \frac{\E\T}{\E\T_c}\to 1, \quad \text{as}\;\; K\to \infty \;(\text{thus}\;\; \alpha\to 0).
\end{align}
Order-$1$ asymptotic optimality proved.
\endproof}
\subsection{Refined Approximations to the Error Probabilities}
Although the asymptotic upper bounds in  Lemma \ref{lemma2} are sufficient to reveal the asymptotic optimality of the CA-DSPRT, their constant terms are not specified in analytical form. Thus the analytical characterization in Lemma \ref{lemma2} offers limited guidance for setting the thresholds $\{A, B\}$ such that  the error probability constraints can be met. To address this limitation, we next provide a refined asymptotic approximations to the error probabilities.  

Defining the difference matrix $\Delta_{t}\triangleq\bW^{t}-\bJ$, then the Type-I error probability can be rewritten as 
\begin{align}\label{refinedError}
\Prob_0\lb D^{(k)}_\text{ca}=1\rb=&\, \Prob_0\lb\be_k^T\lb\sum_{j=1}^{\T_\text{ca}^{(k)}}\bW^{q\lb\T_\text{ca}^{(k)}-j+1\rb}\bs_j\rb\ge B\rb\nonumber\\
=&\, \Prob_0\lb\be_k^T\lb\sum_{j=1}^{\T_\text{ca}^{(k)}}\bJ\bs_j+\sum_{j=1}^{\T_\text{ca}^{(k)}}\Delta_{q\lb\T_\text{ca}^{(k)}-j+1\rb}\bs_j\rb\ge B\rb\nonumber\\=&\,\Prob_0\lb \sum_{j=1}^{\T_\text{ca}^{(k)}}{\bf 1}^T\bs_j+K\be_k^T\sum_{j=1}^{\T_\text{ca}^{(k)}}\Delta_{q\lb\T_\text{ca}^{(k)}-j+1\rb}\bs_j\ge KB\rb.
\end{align}
Note that $\bW$ under Condition \ref{con1} satisfies that $\bW^t\to \bJ$ as $t\to \infty$ \cite{Xiao04}. Drawing on this property, we approximate $\Delta_t\approx {\bf 0}$, for $t> t_0 q$, where $t_0$ can be selected to be sufficiently large according to $\sigma_2(\bW)$ and $q$, and is independent of $A, B$. The smaller $\sigma_2(\bW)$ is, or the greater $q$ is, the faster that $\bW^t$ approaches $\bJ$ and $\Delta_t$ approaches ${\bf 0}$. Applying the Markov inequality to \eqref{refinedError}, we have
\begin{align}\label{approx_alpha}
\Prob_0\lb D^{(k)}_\text{ca}=1\rb&\le {e^{-KB}}\;{\E_0\lb \prod_{j=1}^{\T_\text{ca}^{(k)}}\prod_{\ell=1}^K l_j^{(\ell)}\exp\lb\be_k^T\lb K\sum_{j=1}^{\T_\text{ca}^{(k)}}\Delta_{q\lb\T_\text{ca}^{(k)}-j+1\rb}\bs_j\rb\rb\rb}
\nonumber\\&=e^{-KB}\,{\E_1\lb \exp\lb\be_k^T\lb K\sum_{j=1}^{\T_\text{ca}^{(k)}}\Delta_{q\lb\T_\text{ca}^{(k)}-j+1\rb}\bs_j\rb\rb\rb}
\nonumber\\&\approx e^{-KB}\,{\E_1\lb \exp\lb\be_k^T\lb K\sum_{j=\T_\text{ca}^{(k)}-t_0+1}^{\T_\text{ca}^{(k)}}\Delta_{q\lb\T_\text{ca}^{(k)}-j+1\rb}\bs_j\rb\rb\rb}
\nonumber\\&\approx e^{-KB}\,\underbrace{{\E_1\lb \exp\lb\be_k^T\lb K\sum_{j=1}^{t_0}\Delta_{q j}\bs_j\rb\rb\rb}}_{{\cal C}_\alpha},
\end{align}
where the constant factor ${\cal C}_\alpha$ can be readily computed by simulation since $t_0$ is a prefixed number.  Similarly, we can derive the same approximation to the Type-II error probability:
\begin{align}\label{approx_beta}
\Prob_1\lb D^{(k)}_\text{ca}=0\rb&\le {e^{-KA}}\;{\E_1\lb \prod_{j=1}^{\T_\text{ca}^{(k)}}\prod_{\ell=1}^K 1/l_j^{(\ell)}\exp\lb\be_k^T\lb K\sum_{j=1}^{\T_\text{ca}^{(k)}}\Delta_{q\lb\T_\text{ca}^{(k)}-j+1\rb}\bs_j\rb\rb\rb}
\nonumber\\&=e^{-KA}\,{\E_0\lb \exp\lb\be_k^T\lb K\sum_{j=1}^{\T_\text{ca}^{(k)}}\Delta_{q\lb\T_\text{ca}^{(k)}-j+1\rb}\bs_j\rb\rb\rb}
\nonumber\\&\approx e^{-KA}\,{\E_0\lb \exp\lb\be_k^T\lb K\sum_{j=\T_\text{ca}^{(k)}-t_0+1}^{\T_\text{ca}^{(k)}}\Delta_{q\lb\T_\text{ca}^{(k)}-j+1\rb}\bs_j\rb\rb\rb}
\nonumber\\&\approx e^{-KA}\,\underbrace{{\E_0\lb \exp\lb\be_k^T\lb K\sum_{j=1}^{t_0}\Delta_{q j}\bs_j\rb\rb\rb}}_{{\cal C}_\beta}.
\end{align}
In essence, \eqref{approx_alpha} and \eqref{approx_beta} further specify the constant terms in Lemma \ref{lemma2}, or tighten the constant ${\cal B}_k$ in \eqref{const_proof}. As we will show through the simulations in Section V, these bounds accurately characterize the error probabilities of the CA-DSPRT with proper $t_0$. By the virtue of these refined approximations, the practitioners can determine the thresholds to satisfy the error probability constraints in \eqref{ST} by
\begin{align}
A=-\frac{1}{K}\log \frac{\beta}{{\cal C}_\beta}, \;\;\text{and}\;\; B=-\frac{1}{K}\log \frac{\alpha}{{\cal C}_\alpha},
\end{align}
which considerably simplifies the thresholds selection for the CA-DSPRT.
\ignore{Let us assume
\begin{align}
\bEta_t=\bP\lb\bEta_{t-1}+\bs_t\rb,\quad \bEta_0=\bP\bs_0=\bP\by
\end{align}}

\ignore{\subsection{Exact Characterization}
We begin from the false alarm probability. Denote the event $\cE\triangleq\{\bS_t: \be_\ell^T\bS_t\in (-A,B)\}$, the indicator function ${I}(\cE)\triangleq \mathbbm{1}_{\{\cE\}}$, and the false alarm probability under $\bS_0=\by$ as
\begin{align}
 h(\by)\triangleq\Prob_0^\by\lb D^i_{\T_i}=1\rb=\E_0^\by\lb I\lb \be_\ell^T\bS_{\T^{(\ell)}}\ge B\rb\rb.
\end{align}
Then we have the exact characterizations 
\begin{align}\label{Eq1}
  h(\by)&= \E_0^\by\lb I\lb \be_\ell^T\bS_{\T^{(\ell)}}\ge B\rb I\lb\T^{(\ell)}=1\rb\rb +\E_0^\by\lb I\lb \be_\ell^T\bS_{\T^{(\ell)}}\ge B\rb I\lb\T^{(\ell)}>1\rb\rb
\end{align}
\begin{align}
  \text{First term of the rhs of \eqref{Eq1}}&=\E_0^\by\lb I\lb \be_\ell^T\bS_1\ge B\rb\rb \nonumber\\&=\underbrace{\E_0^\by \left[ I\lb \be_\ell^T\bW\lb\by+\bs_{1}\rb\ge B \rb\right]}_{u(\by)}
\end{align}

\begin{align}
  \text{Second term of the rhs of \eqref{Eq1}}&=\E_0^\by\lb I\lb \be_\ell^T \bS_{\T^{(\ell)}}\ge B\rb I\lb \be_\ell^T\bS_1\in(-A,B)\rb\rb \nonumber\\&=\int_{\bS_1\in \cE}\E_0^\by\lb\left. I\lb \be_\ell^T\bS_{\T^{(\ell)}}\ge B\rb \right| \bS_1\rb \Prob_0\lb d\bS_1\rb
\end{align}
Then we have the exact characterization for error proability:
\begin{align}
  h(\by)=u(\by)+\int_{\bx\in \cE} h(\bx)\;f_{\bW\by+\bs_1}(\bx)d\bx.
\end{align}
{Here we can derive a necessary condition for the consensus-based method to yield asymptotic optimality. The condition should be on the tail distribution of the individual statistic. At least, the necessary condition is that $u(0)\sim e^{-KB}$, which leads to regularity for the distribution of the statistics.}

\begin{align}
h(\by)=\int \left[I\lb \be_\ell^T\bx\ge B\rb+I\lb -A\le \be_\ell^T\bx\le B\rb h(\bx)\right]f_{\bW\by+\bs_1}(\bx)d\bx
\end{align}
Let us assume 
\begin{align}
f_{\bW\by+\bs_1}(\bx)=g(\by)\phi(\bx).
\end{align}
Then 
\begin{align}
h(\by)=g(\by)\underbrace{\int \left[I\lb \be_\ell^T\bx\ge B\rb+I\lb -A\le \be_\ell^T\bx\le B\rb h(\bx)\right]\phi(\bx)d\bx}_{\eta(A, B)}
\end{align}
Thus we have
\begin{align}
h(\by)&=g(\by)\int \left[I\lb \be_\ell^T\bx\ge B\rb+I\lb -A\le \be_\ell^T\bx\le B\rb\eta(A,B)g(\bx)\right]\phi(\bx)d\bx\\&=g(\by)\int \left[I\lb \be_\ell^T\bx\ge B\rb+I\lb -A\le \be_\ell^T\bx\le B\rb\frac{h(\by)}{g(\by)}g(\bx)\right]\phi(\bx)d\bx\\&=g(\by)\underbrace{\int I\lb \be_\ell^T\bx\ge B\rb \phi(\bx) \;d\bx}_{c_0(B)}+h(\by)\underbrace{\int I\lb -A\le \be_\ell^T\bx\le B\rb g(\bx) \phi(\bx) \; d\bx}_{c_1(A,B)}
\end{align}
leading to 
\begin{align}
h(\by)=\frac{c_0(B)g(\by)}{1-c_1(A,B)}.
\end{align}
{\color{red} Closed-form approximation; numerical approximation. }
\subsection{Exponential Samples}
\begin{align}
f_{\bW\by+\bs}\lb\bx\rb=\prod_{i=1}^K \lambda_i\exp\lb-\lambda_i x_i +\lambda_i\lb\bW\by\rb_i\rb=\underbrace{\prod_{i=1}^K\lambda_i\exp\lb -\lambda_ix_i\rb}_{\phi(\bx)}\underbrace{\prod_{i=1}^K\exp\lb \lambda_i\lb \bW\by\rb_i\rb}_{g(\by)}
\end{align}}

\section{Numerical Results}
In this section, we examine the performance of the two message-passing-based distributed sequential tests using two sample distributions. Extensive numerical results will be provided to corroborate the theoretical results developed in this work. 

We begin by deciding the weight matrix for the consensus algorithm. There are multiple methods to choose $\bW$ such that Condition \ref{con1} can be satisfied, one of which is assigning equal weights to the data from neighbours \cite{Xiao04,Sahu16}. In specific, the message-passing protocol \eqref{protocol_2} becomes
\begin{align}
\widetilde\eta^{(k)}_{t, m}&=\lb1-|{\cal N}_k|\delta\rb\,\widetilde\eta^{(k)}_{t, m-1}+\delta\sum_{\ell\in {\cal N}_k}\, \widetilde \eta^{(\ell)}_{t, m-1},\nonumber\\&=\widetilde\eta^{(k)}_{t, m-1}+\delta\sum_{\ell\in {\cal N}_k}\, \lb\widetilde \eta^{(\ell)}_{t, m-1}-\widetilde\eta^{(k)}_{t, m-1}\rb
\;\; \text{for}\;\; m=1, 2, \ldots, q.
\end{align}
As such, the weight matrix admits
\begin{align}\label{EqualM}
\bW={\bf I}-\delta\big(\underbrace{\bD-\bA}_{\bL}\big),
\end{align}
where ${\bA}$ is the adjacent matrix, whose entries $a_{i,j}=1$ if and only if $\{i, j\}\in {\cal E}$,  and $\bD\triangleq \text{diag}\left\{|{\cal N}_1|, |{\cal N}_2|, \ldots, |{\cal N}_K|\right\}$ is the called the degree matrix. Their difference is called the Laplacian matrix $\bL$ which is positive semidefinite. First, $\bW{\bf 1}={\bf 1}$ and ${\bf 1}^T\bW={\bf 1}^T$ hold for any value of $\delta$ due to the definition of $\bL$ (i.e., $\bL {\bf 1}={\bf 0}$ and ${\bf 1}^T\bL={\bf 0}^T$). Second, note that $\bW$ in \eqref{EqualM} is a symmetric matrix, whose second largest singular value $$\sigma_2\lb\bW\rb=\max\left\{1-\delta\lambda_{n-1}\lb\bL\rb, \delta\lambda_1\lb \bL\rb-1\right\}<1,$$ if and only if $0<\delta<\frac{2}{\lambda_{1}\lb\bL\rb}$. Within this interval, we set $\delta=\frac{2}{\lambda_1(\bL)+\lambda_{n-1}(\bL)}$ such that the constant terms in  Theorem \ref{thm1} are minimized, or equivalently, $\sigma_2(\bW)$ is minimized. Condition \ref{con2} on the LLR distribution will be verified for the particular testing problem in Section VI-A and B respectively.

In the following experiments, we consider a specific class of network topology as an example, where each sensor is connected to sensors within $m$ links, as denoted as $\cG(n,m)$. 
For instance, in $\cG(12,2)$ illustrated in Fig. \ref{Network}, each sensor is connected to the sensors within range $2$.

\begin{figure}
\centering
\includegraphics[width=0.45\columnwidth]{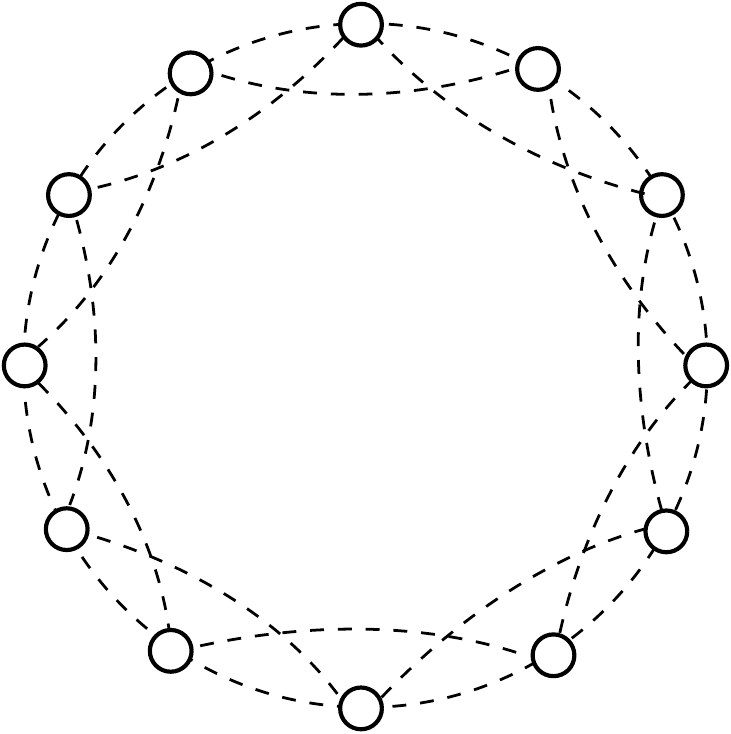}
\caption{The sensor network represented by a graph $\cG(12,2)$.}\label{Network}
\end{figure}

\subsection{Gaussian Samples}
First we consider the problem of detecting the mean-shift of Gaussian samples. Without loss of generality, the variance is assumed to be one in the hypothesis testing problem, i.e., 
\begin{align*}
&\cH_0: X^{(k)}_t\sim {\cal N}(0, 1),\\
&\cH_1: X^{(k)}_t\sim {\cal N}(\mu, 1),\;\; k=1, 2, \ldots, K,\;\; t=1, 2, \ldots
\end{align*}
The LLR at sensor $k$ is given by
\begin{align}\label{LLR_Gaussian}
s_t^{(k)}=X_t^{(k)}\mu-\frac{\mu^2}{2}\sim \left\{
\begin{array}{ll}
{\cal N}\lb -\frac{\mu^2}{2},\mu^2\rb, & \text{under }{\cal H}_0,\\
{\cal N}\lb \frac{\mu^2}{2},\mu^2\rb, &  \text{under }{\cal H}_1,
\end{array}\right.
\end{align}
with KLDs equal to
$$\cD_0^{(k)}=\cD_1^{(k)}=\frac{\mu^2}{2}.$$ 
Note that
\begin{align}
&\E_0\lb e^{K\sqrt{K}|s_t^{(k)}|}\rb=\E_1\lb e^{K\sqrt{K}|s_t^{(k)}|}\rb \nonumber\\&=e^{\lb K\sqrt{K}+1\rb K\sqrt{K}\mu^2/2}\Phi\lb \lb K\sqrt{K}+\frac{1}{2}\rb\mu\rb +e^{\lb K\sqrt{K}-1\rb K\sqrt{K}\mu^2/2}\Phi\lb \lb K\sqrt{K}-\frac{1}{2}\rb\mu\rb
\end{align}
turns out to be a constant, thus the LLR \eqref{LLR_Gaussian} satisfies the Condition 2. As a result, the CA-DSPRT achieves the order-$2$ asymptotically optimal performance at every sensor. Moreover, for comparison, we will also plot the analytical bounds derived in \cite{Sahu16} for the error probabilities of the CA-DSPRT with $q=1$, i.e., 
\begin{align}\label{Sahu_alpha}
\Prob_0\lb D^{(k)}_\text{ca}=1\rb\le \frac{2\exp\lb -\frac{\sigma_2(\bW)KB}{8\lb K\sigma_2(\bW)^2+1\rb}\rb}{1-\exp\lb -\frac{K\cD_1}{4\lb K\sigma_2(\bW)^2+1\rb}\rb},
\end{align}
and for the stopping time characterization, i.e., 
\begin{align}\label{Sahu_ET}
\E_i\lb\T^{(k)}_\text{ca}\rb\le\frac{10\lb K\sigma_2^2(\bW)+1\rb}{7} \E_i\lb\T_\text{c}\rb, \quad i =0, 1,
\end{align}
given the same error probabilities. They are referred to as the existing analysis for the CA-DSPRT. Note that the analysis in \cite{Sahu16} does not reveal the asymptotic optimality of $\T_\text{ca}^{(k)}$.
\begin{figure}
\centering
\subfigure[]{\includegraphics[width=0.8\textwidth]{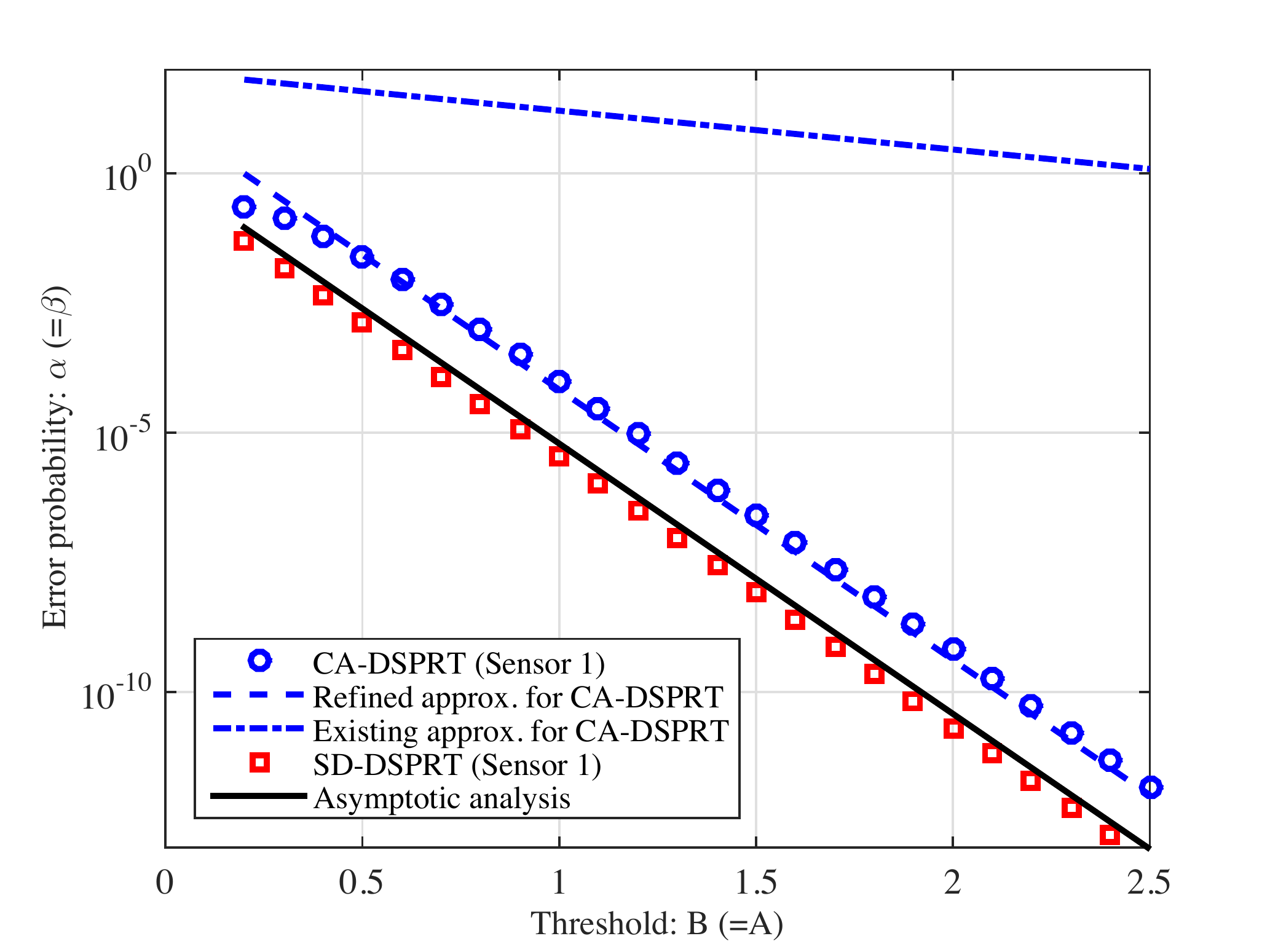}}
\subfigure[]{\includegraphics[width=0.8\textwidth]{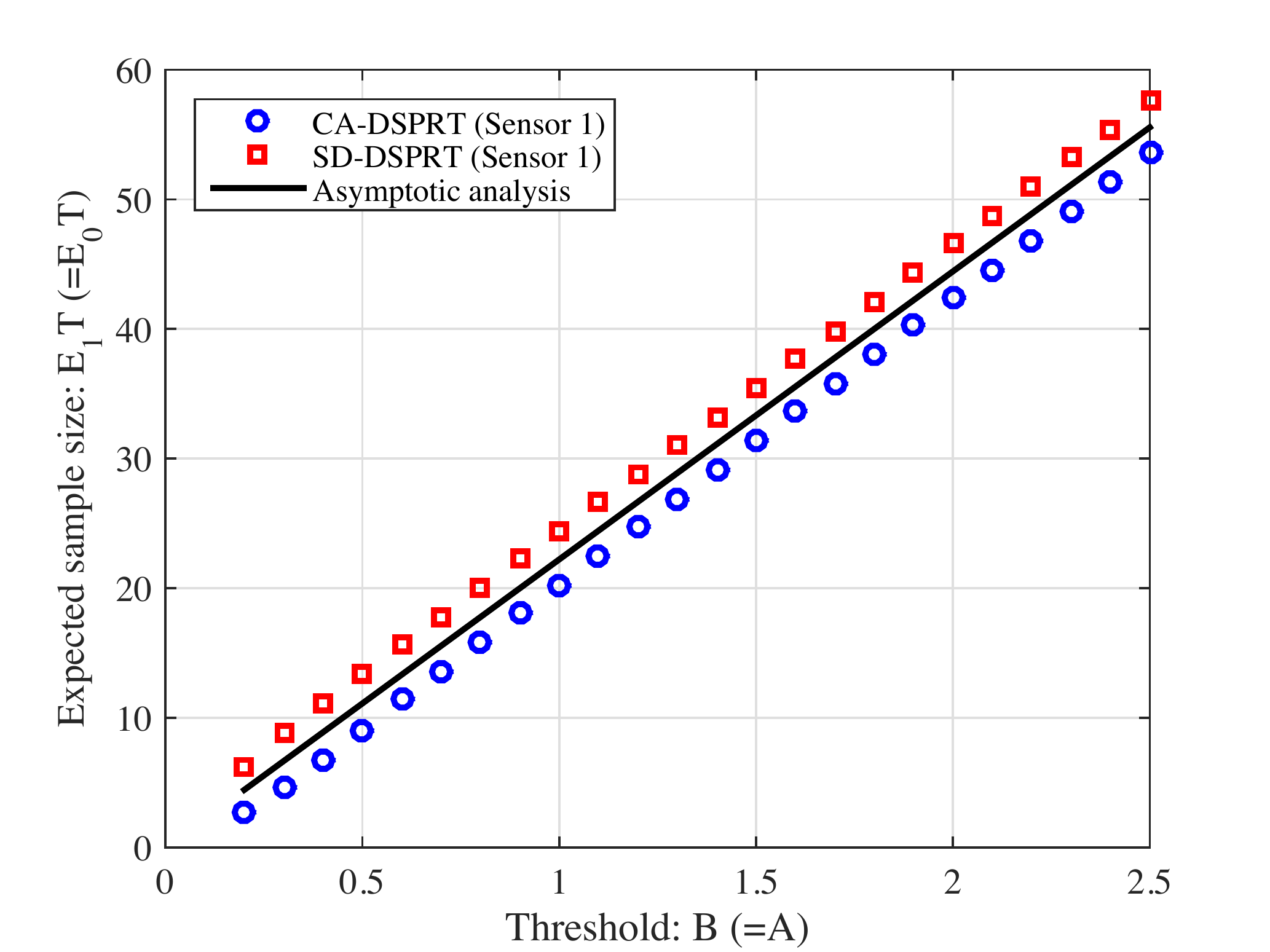}}
\caption{The false alarm probability and expected sample size in terms of the threshold $B$ for the network $\cG\lb 12,2\rb$. }\label{C12_2_1}
\end{figure}
\begin{figure}
\centering
\includegraphics[width=0.86\textwidth]{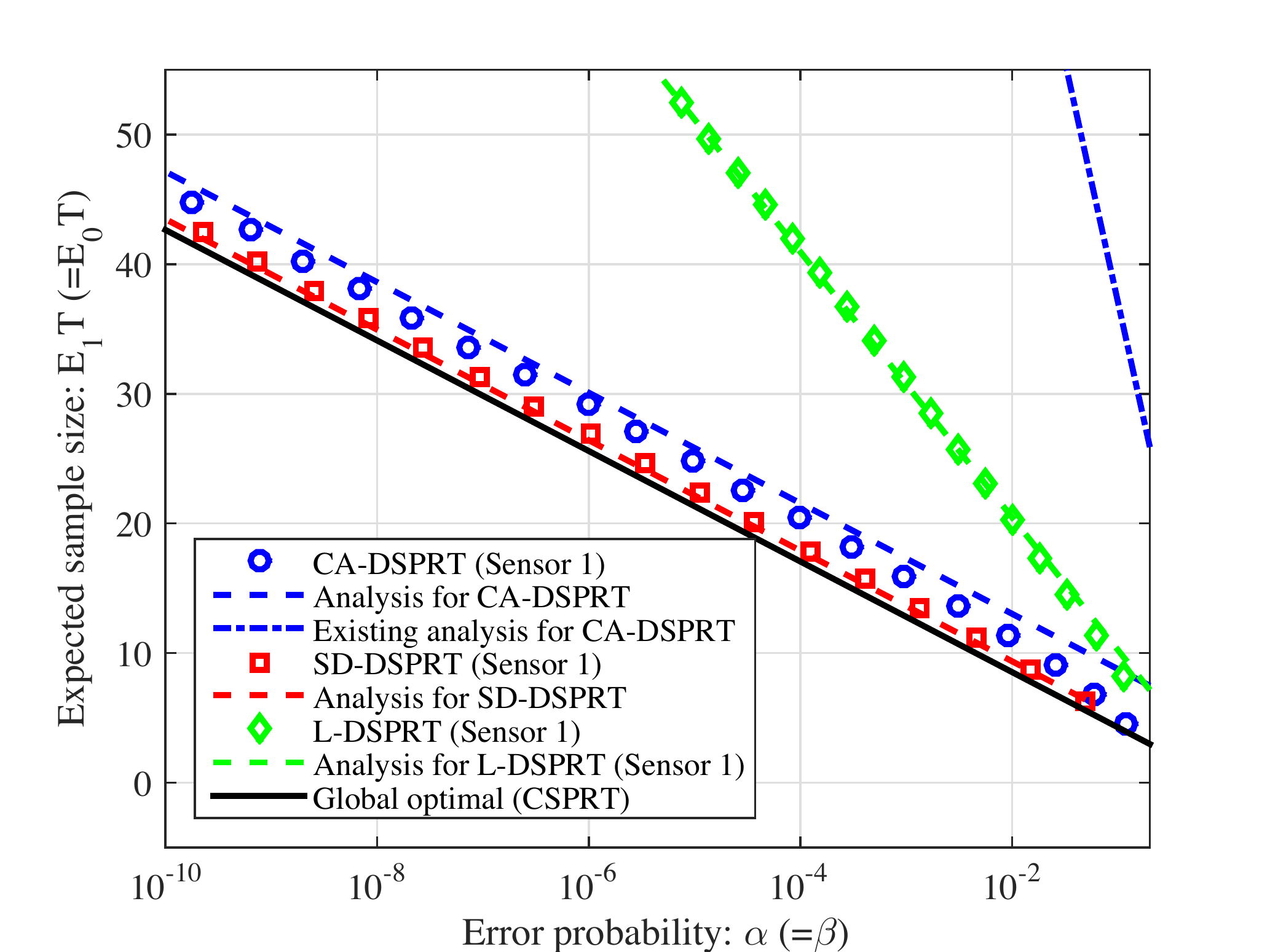}
\caption{Stopping time performances of different message-passing-based distributed sequential tests for the network $\cG\lb 12,2\rb$.}\label{C12_2_2}
\end{figure}

Since sensors in the network have the identical sample distributions, identical adjacent sensors and message-passing weights, they should result in identical test performance under SD-DSPRT, CA-DSPRT and L-DSPRT respectively. Thus, we only plot the performance at sensor $1$ for illustrative purpose, bearing in mind that the performance at other sensors align identically to that of sensor $1$.  In addition, due to the symmetry of the statistic distribution under $\cH_0$ and $\cH_1$, it is sufficient to plot  the performance under one hypothesis, while the other follows identically. Accordingly, we demonstrate the false alarm probability $\alpha$ and expected sample size $\E_1\lb\T\rb$ henceforth.

Let us first consider the sensor network $\cG\lb 12, 2\rb$ as depicted in Fig. \ref{Network} whose weight matrix \eqref{EqualM} has $\sigma_2\lb\bW\rb=0.6511$. The alternative mean is set as $\mu=0.3$. The number of message-passings for the CA-DSPRT at each time slot is fixed as $q=1$. Fig. \ref{C12_2_1} illustrates how the error probability and expected sample size change with the threshold in SD-DSPRT and CA-DSPRT. Specifically, Fig. \ref{C12_2_1}-(a) shows that the error probability of the SD-DSPRT (marked in red squares) is the same as that of the CSPRT (marked in black solid line), i.e., $e^{-KB}$, while that of the CA-DSPRT (marked in blue circles) aligns parallel to the solid line, as expected by Lemma \ref{lemma2}. Moreover, the refined approximation \eqref{approx_alpha} accurately characterizes the error probability with $t_0=10$ whereas  the curve by \eqref{Sahu_alpha} deviates far away from the simulation result. Fig. \ref{C12_2_1}-(b) shows that the expected sample sizes of SD-DSPRT and CA-DSPRT align parallel to that of the CSPRT as the threshold increases, which agrees with  \eqref{MI_ET1} and Lemma \ref{lemma1}. 

Combining \ref{C12_2_1}-(a) and (b) gives the performance curves as shown in Fig. \ref{C12_2_2}.  First, both the performances of SD-DSPRT and CA-DSPRT only deviate from the global optimal performance by a constant margin as $A, B\to \infty$, exhibiting the order-$2$ asymptotic optimality as stated in Theorems \ref{thm0} and \ref{thm1}. Particularly, SD-DSPRT shows relatively smaller degradation compared to the CA-DSPRT. However, this superiority is gained at the cost of substantially heavier communication overhead. In addition, we also plot the performance of L-SRPRT (marked in green diamonds), which is clearly seen to be sub-optimal and diverges from the optimal performance by orders of magnitude. The curve by \eqref{Sahu_ET} again substantially deviates from the true performance.

Another experiment is demonstrated in Figs. \ref{C20_2_1} and \ref{C20_2_2} based on the network $\cG(20,2)$ with $\sigma_2\lb\bW\rb=0.8571$. It is seen that our analyses still accurately characterize the performances of SD-DSPRT and CA-DSPRT in the asymptotic regime where $A, B\to\infty$ and $\alpha, \beta\to 0$. Note that for $q=1$, the constant gap between the CA-DSPRT and CSPRT is greater compared to the preceding simulation due to a larger $\sigma_2\lb\bW\rb$. Interestingly and expectedly, if we  increase the number of message-passings by one, i.e., $q=2$, the constant gap between the CA-DSPRT and CSPRT can be substantially reduced. This implies that, in practice, we can control the number of message-passings in the consensus algorithm to push the CA-DSPRT closer to the global optimum. Nevertheless, changing $q$ only varies the constant gap; in any case, the order-$2$ asymptotic optimality of the SD-DSPRT and CA-DSPRT are clearly seen in Fig. \ref{C20_2_2}. 
\begin{figure}
\centering
\subfigure[]{\includegraphics[width=0.8\textwidth]{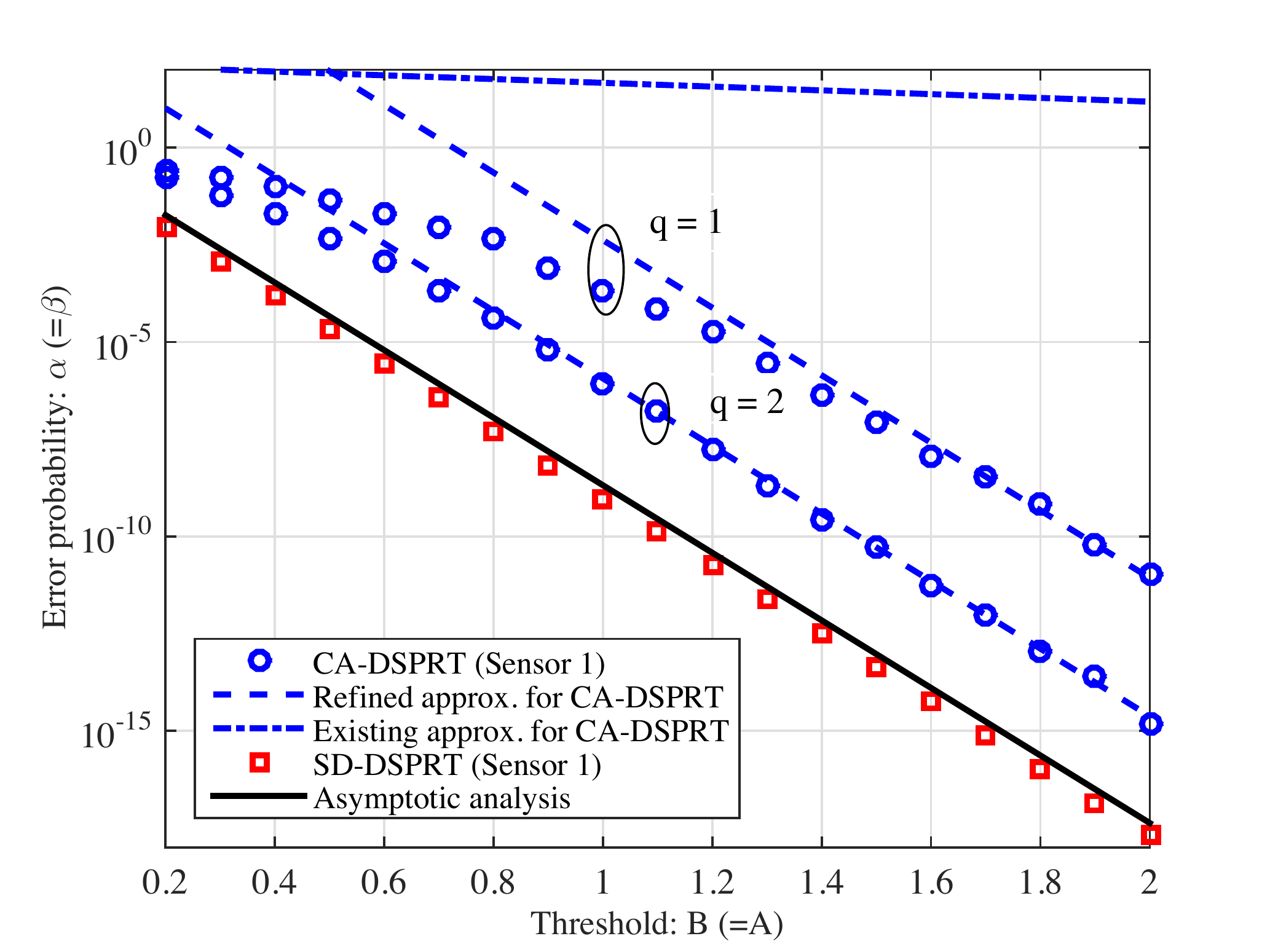}}
\subfigure[]{\includegraphics[width=0.8\textwidth]{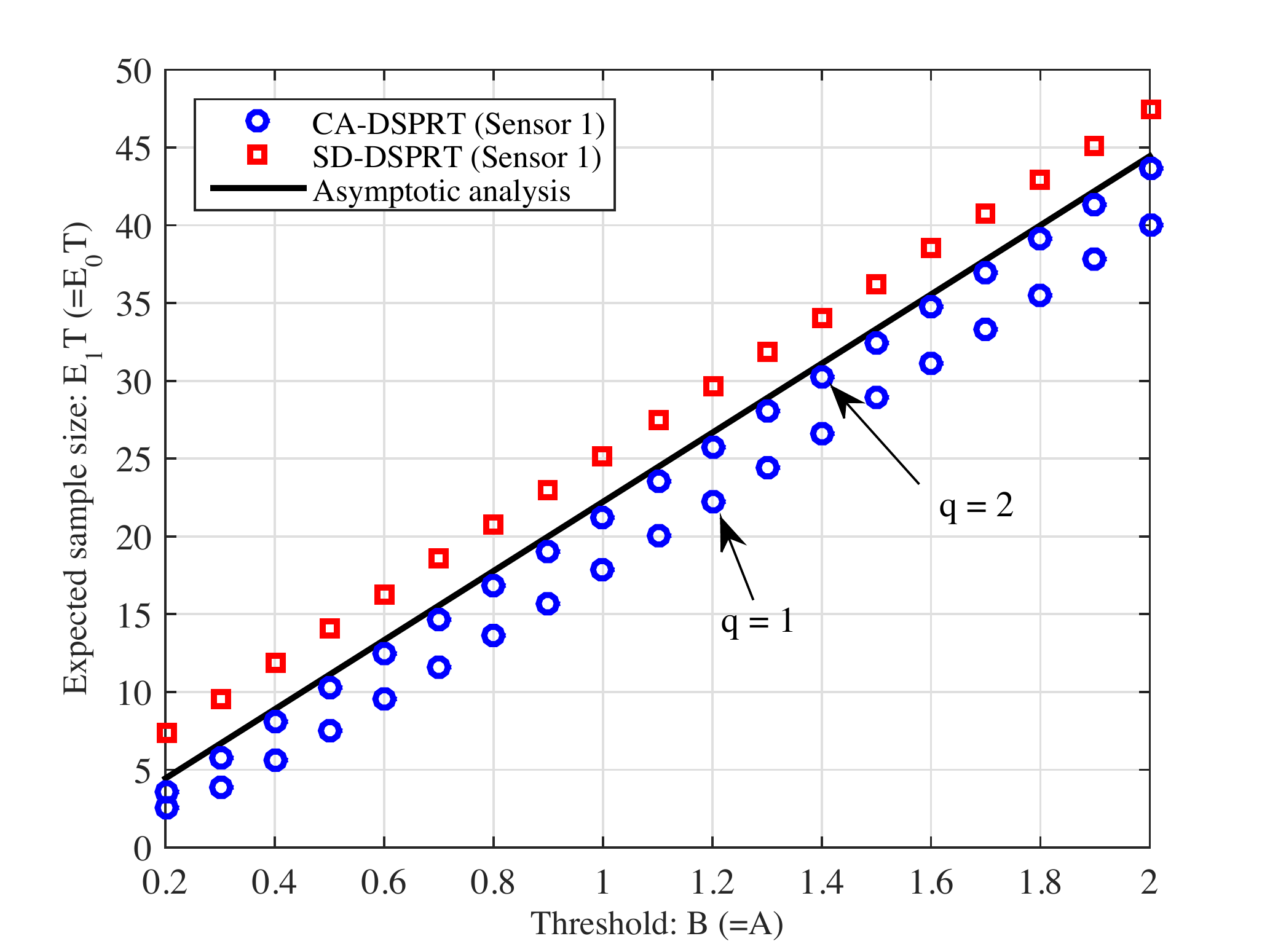}}
\caption{The false alarm probability and expected sample size in terms of the threshold $B$ for the network $\cG\lb 20,2\rb$. }\label{C20_2_1}
\end{figure}
\begin{figure}
\centering
\includegraphics[width=0.86\textwidth]{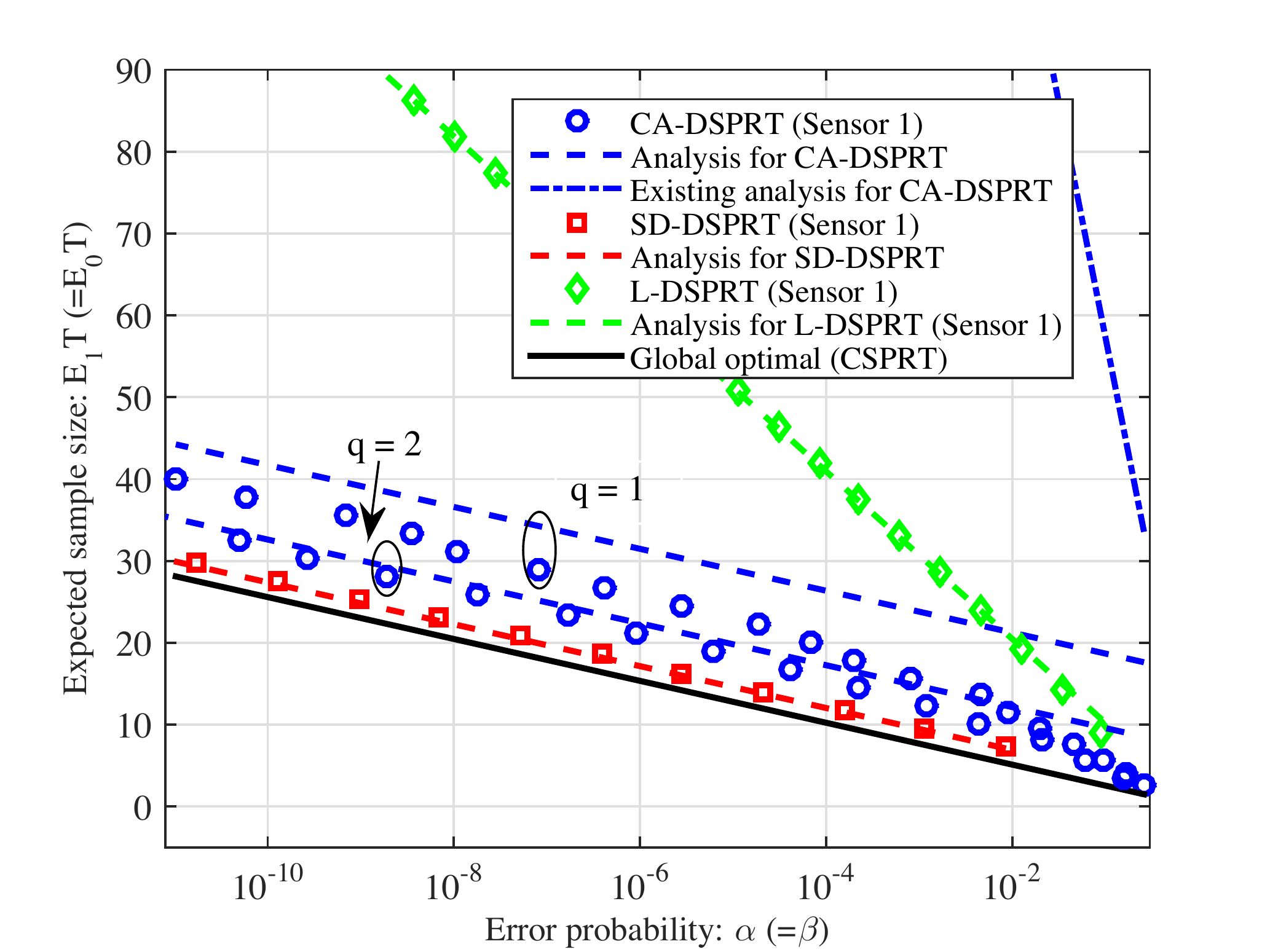}
\caption{Stopping time performances of different message-passing based distributed sequential tests for the network $\cG\lb 20, 2\rb$. }\label{C20_2_2}
\end{figure}
\subsection{Laplacian Samples}
Next we apply the message-passing-based distributed sequential tests to detect the mean-shift of the Laplace samples, whose the dispersion around the mean is wider than the Gaussian samples. Laplace distribution is widely used for modelling the data with heavier tails, with applications in speech recognition, biological process analysis, and credit risk prediction in finance. Without loss of generality, we assume $b=1$ for the probability density function $f\lb x\rb=\frac{1}{2b}\exp\lb -\frac{|x-\mu|}{b}\rb$, i.e., 
\begin{align*}
&\cH_0: X^{(k)}_t\sim \text{Laplace}\lb 0,{1}\rb, \\
&\cH_1: X^{(k)}_t\sim \text{Laplace}\lb \mu,{1}\rb, \quad k=1, 2, \ldots, K,\;\; t=1, 2, \ldots,
\end{align*}
with the LLR at sensor $k$ given by
\begin{align}
s_t^{(k)}=\left\{
\begin{array}{ll}
\mu & X_t^{(k)}<0,\\
2X_t^{(k)}-\mu & 0\le X_t^{(k)}\le \mu,\\
 -\mu & X_t^{(k)}\ge \mu,
\end{array}\right.
\end{align}
and KLDs equal to
\begin{align}
\cD_0^{(k)}=\cD_1^{(k)}=|\mu|-1+e^{-|\mu|}.
\end{align}
Under this problem setting, Condition \ref{con2} is easily verified by noting that $|s_t^{(k)}|$ is bounded above by $\mu$, thus $\E_i\lb e^{K\sqrt{K}|s_t^{(k)}|}\rb$ is bounded above by constant $e^{K\sqrt{K}\mu}$.
\begin{figure}
\centering
\subfigure[]{\includegraphics[width=0.8\textwidth]{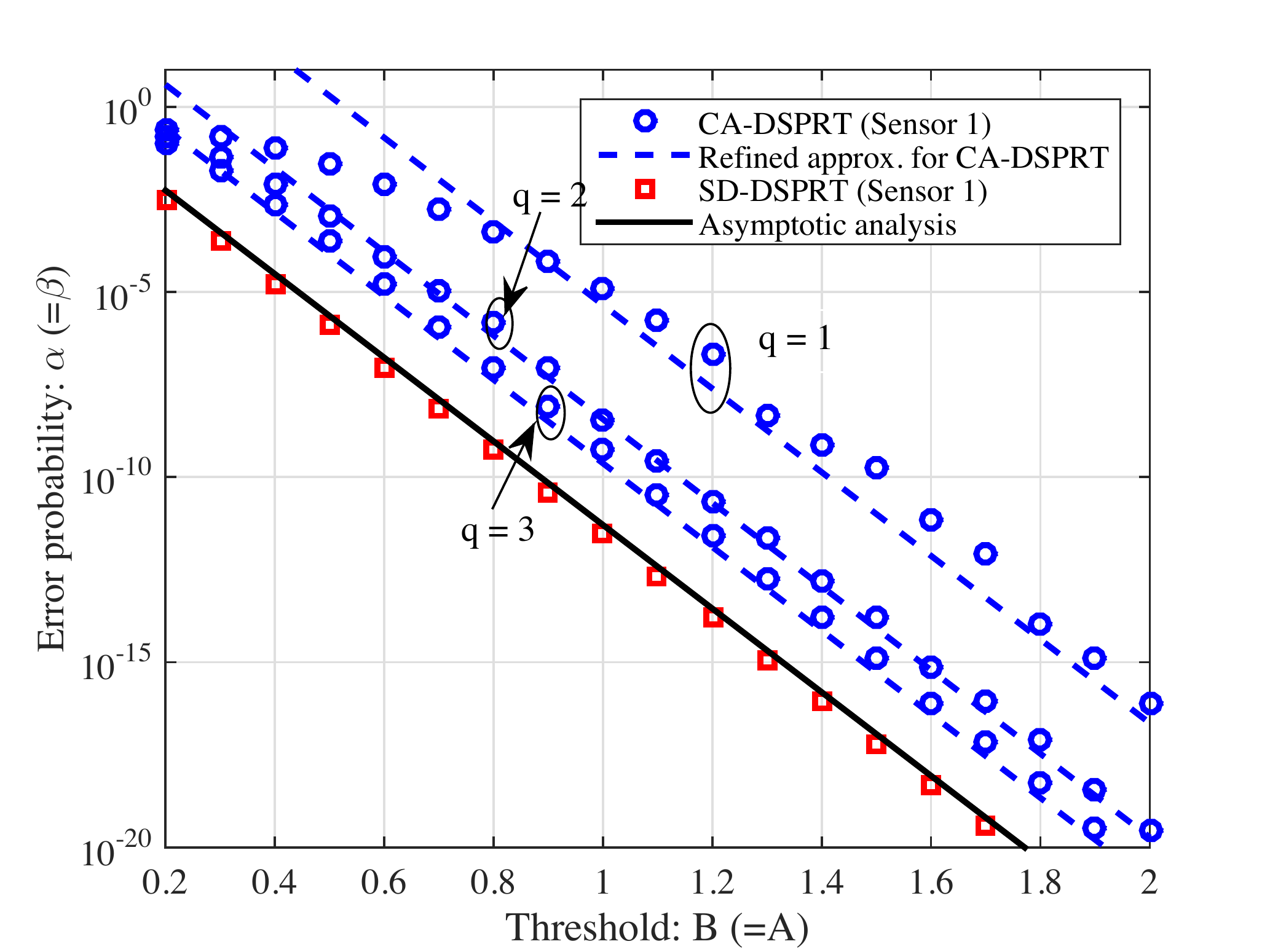}}
\subfigure[]{\includegraphics[width=0.8\textwidth]{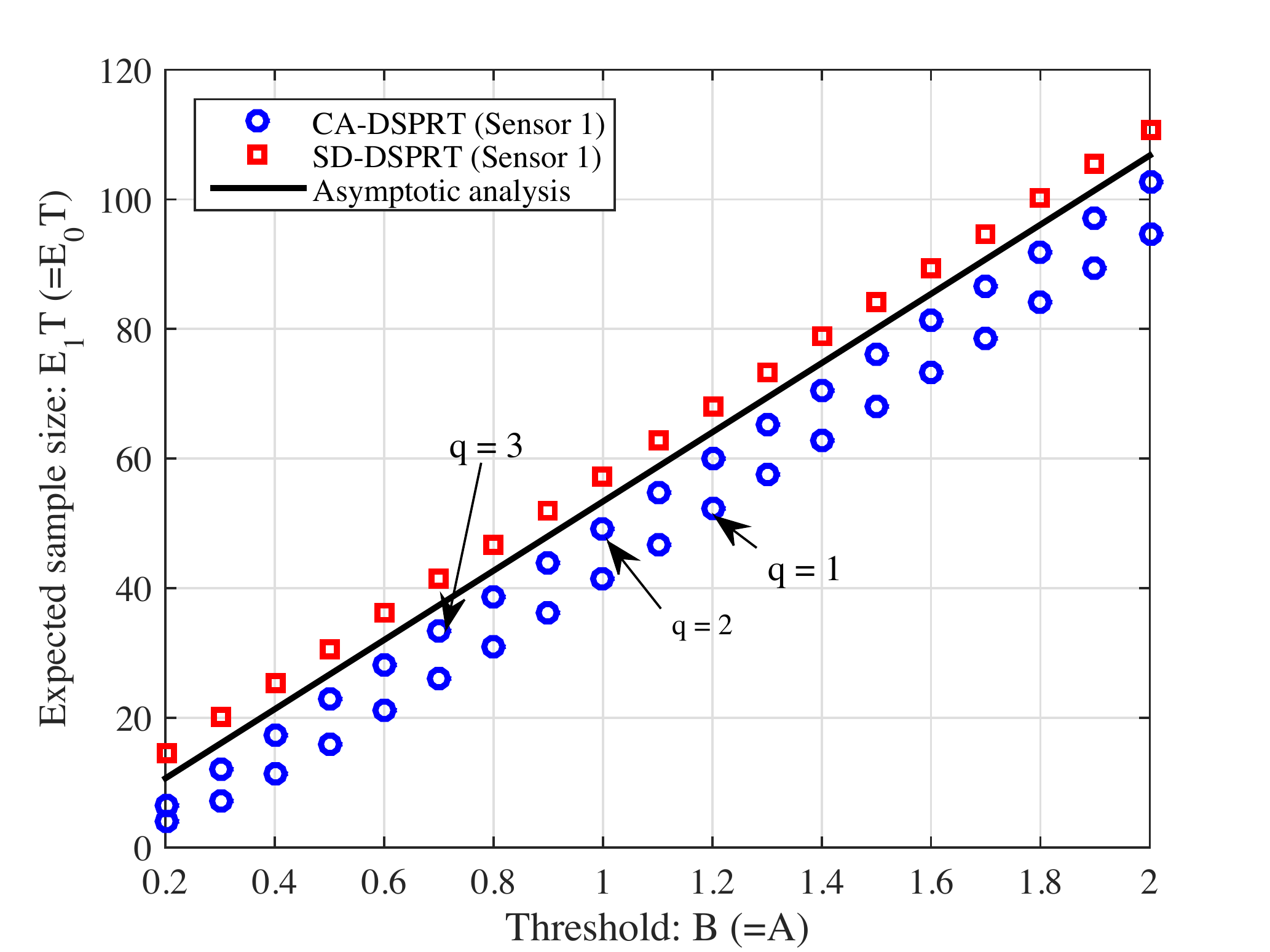}}
\caption{The false alarm probability and expected sample size in terms of the threshold $B$ for the network $\cG\lb 26,2\rb$. }\label{C26_2_1}
\end{figure}
\begin{figure}
\centering
\includegraphics[width=0.86\textwidth]{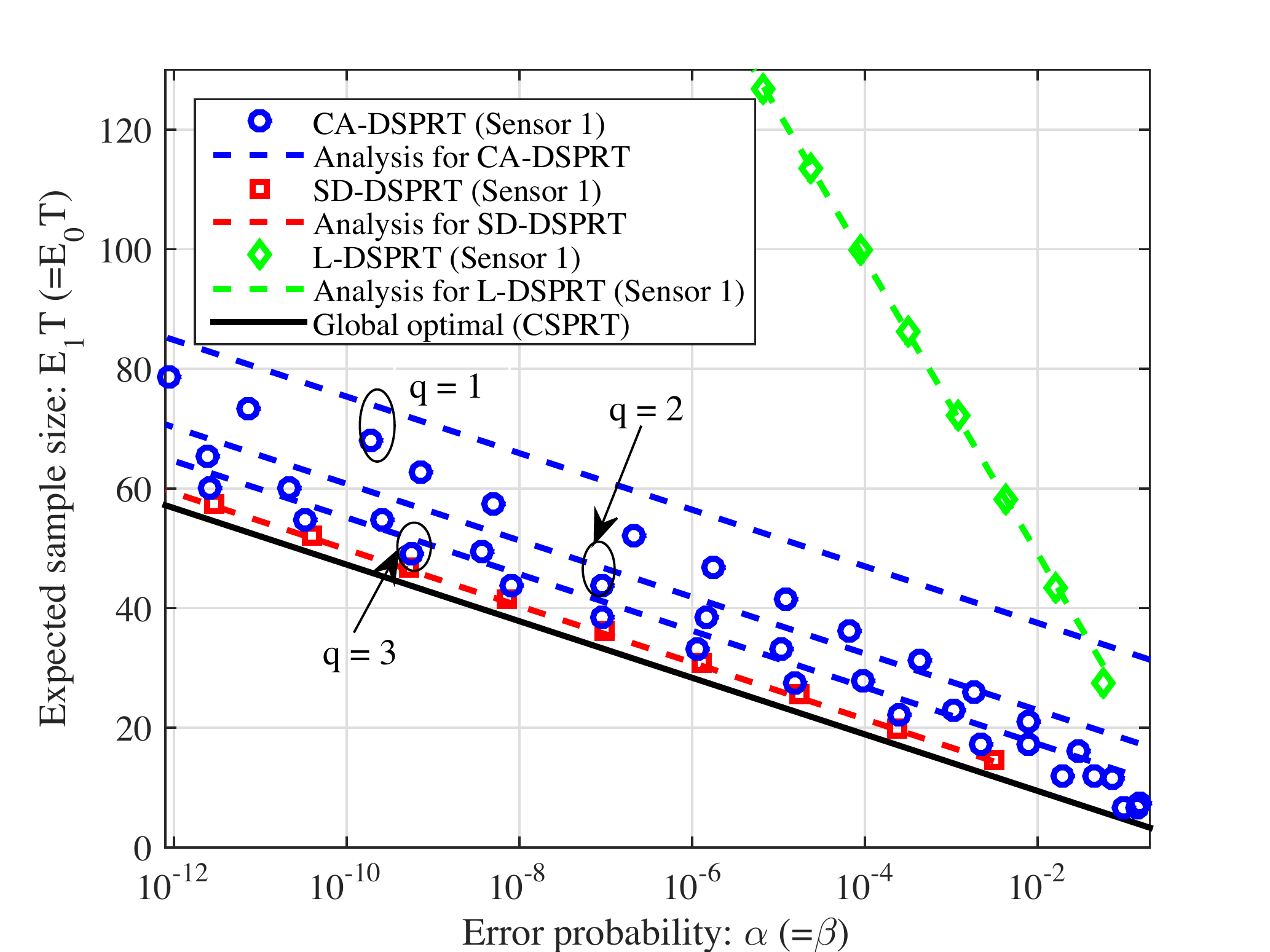}
\caption{Stopping time performances of different message-passing based distributed sequential tests for the network $\cG \lb 26, 2\rb$.}\label{C26_2_2}
\end{figure}

We consider the network $\cG \lb 26, 2\rb$ whose weight matrix \eqref{EqualM} has $\sigma_2\lb\bW\rb=0.9115$, and the alternative mean is fixed as $\mu=0.2$. In Fig. \ref{C26_2_1}-(a), the error probability of the SD-DSPRT is the same as that given by the  asymptotic analysis, i.e., $e^{-KB}$, while that of the CA-DSPRT stays parallel to the asymptotic result. Similarly, the expected sample sizes shown in Fig. \ref{C26_2_1}-(b) also agree with the asymptotic analysis. Again, slightly increasing $q$ is seen to quickly  narrow down the constant gaps.  In Fig. \ref{C26_2_2}, both SD-DSPRT and CA-DSPRT (for any value of $q$) deviate from the global optimal performance by a constant margin as the error probabilities go to zero. In particular,  the CA-DSPRT becomes nearly the same  as the SD-DSPRT for $q=3$, with much less communication overhead. In contrast, the naive L-DSPRT substantially diverges from the global optimum for small error probability.


\section{Conclusions}
In this work, we have investigated the fully distributed sequential hypothesis testing, where each sensor performs the sequential test while exchanging information with its adjacent sensors. Two message-passing-based schemes have been considered. The first scheme hinges on the dissemination of the data samples over the network, and we have shown that it achieves the order-$2$ asymptotically optimal performance at all sensors. However, the dissemination of data samples across the network becomes impractical as the network size grows. In contrast, the second scheme builds on the well-known consensus algorithm, that only requires the exchange of local decision statistic, thus requiring significantly lower communication overhead. We have shown that the consensus-algorithm-based distributed sequential test also achieves the order-$2$ asymptotically optimal performance at every sensor. 
Several future directions can be pursued. First, one can improve the SD-DSPRT by introducing more efficient sample dissemination scheme. Second, note that Condition \ref{con1} on the network topology is in fact more strict than that given in \cite{Xiao04}. It would be interesting to investigate whether the same condition in \cite{Xiao04} can lead to the asymptotic optimality of the CA-DSPRT. It is also of interest to integrate the quantized consensus algorithm into the distributed sequential test, where local decision statistics are quantized into finite bits before message-passing. Moreover, it is practically and theoretically interesting to study the effect of  the time-varying network topology and link failures on  the distributed sequential test. Last but not least, it is of interest to consider fully distributed sequential change-point detection and its asymptotic property. 


\bibliographystyle{IEEEtran}
\bibliography{IEEEabrv,references}
\end{document}